\RequirePackage[l2tabu, orthodox]{nag} 
\documentclass[11pt
]{article} 
\usepackage[paperwidth=8.45in, paperheight=11in, top=1.25in, bottom=1.25in, left=1.00in, right=1.00in]{geometry} 
\usepackage{slashbox}
\usepackage{imakeidx} 
	\makeindex

\usepackage{amssymb,amsmath,amsthm}

\usepackage{mathtools}
\mathtoolsset{showonlyrefs}
\usepackage[dvipsnames]{xcolor} 
\usepackage{tikz}
\usepackage{tikz-3dplot}
\usepackage[latin1]{inputenc} 
\usepackage[T1]{fontenc}
\usepackage{lmodern}
\usepackage[justification=RaggedRight]{caption}
\usepackage[all,error]{onlyamsmath} 

\makeatletter
\newsavebox\myboxA
\newsavebox\myboxB
\newlength\mylenA
\newcommand*\xoverline[2][0.75]{%
    \sbox{\myboxA}{$\m@th#2$}%
    \setbox\myboxB\null
    \ht\myboxB=\ht\myboxA%
    \dp\myboxB=\dp\myboxA%
    \wd\myboxB=#1\wd\myboxA
    \sbox\myboxB{$\m@th\overline{\copy\myboxB}$}
    \setlength\mylenA{\the\wd\myboxA}
    \addtolength\mylenA{-\the\wd\myboxB}%
    \ifdim\wd\myboxB<\wd\myboxA%
       \rlap{\hskip 0.5\mylenA\usebox\myboxB}{\usebox\myboxA}%
    \else
        \hskip -0.5\mylenA\rlap{\usebox\myboxA}{\hskip 0.5\mylenA\usebox\myboxB}%
    \fi}
\makeatother

\newcommand{\Xibar}{\xoverline[1.3]{\Xi}}
\newcommand{\F}{\mathcal{F}}

\renewcommand{\P}{\mathbb{P}}
\newcommand{\Q}{\mathbb{Q}}
\newcommand{\E}{\mathbb{E}}
\DeclareMathOperator{\Var}{\mathbb{V}ar}

\newcommand{\R}{\mathbb{R}}

\newcommand{\ofp}{(\Omega, \mathcal{F}, \P)}
\newcommand{\offp}{(\Omega, \mathcal{F}, \mathbb{F}, \P)}
\newcommand{\lnofp}{L^0(\Omega, \mathcal{F}, \P)}
\newcommand{\loofp}{L^1(\Omega, \mathcal{F}, \P)}
\newcommand{\lpofp}{L^p(\Omega, \mathcal{F}, \P)}
\newcommand{\lphiofp}{L^\Phi(\Omega, \mathcal{F}, \P)}
\newcommand{\lnpofp}{L^0_+(\Omega, \mathcal{F}, \P)}

\newcommand{\conv}{\mathrm{conv}}
\newcommand{\convc}{\overline{\mathrm{conv}}}

\DeclareMathOperator*{\argmin}{arg \, min}
\DeclareMathOperator*{\argmax}{arg \, max}

\def\ind{{\mathchoice{1\mskip-4mu\mathrm l}{1\mskip-4mu\mathrm l}
{1\mskip-4.5mu\mathrm l}{1\mskip-5mu\mathrm l}}}

\newtheorem{theorem}{Theorem}[section]
\newtheorem{definition}[theorem]{Definition}
\newtheorem{corollary}[theorem]{Corollary}
\newtheorem{proposition}[theorem]{Proposition}
\newtheorem{remark}[theorem]{Remark}
\newtheorem{lemma}[theorem]{Lemma}
\newtheorem{example}[theorem]{Example}
\newtheorem{assumption}[theorem]{Assumption}

\usepackage[normalem]{ulem}

\renewcommand{\le}{\leqslant}
\renewcommand{\leq}{\leqslant}
\renewcommand{\ge}{\geqslant}
\renewcommand{\geq}{\geqslant}

\title{\textbf{Cost-efficiency in Incomplete Markets}}

\author{
Carole Bernard \thanks{%
C. Bernard, Grenoble Ecole de Management,  12~rue Pierre S\'{e}mard, 38000 Grenoble, France  (e-mail: 
\texttt{carole.bernard@grenoble-em.com}).}
\thanks{Vrije Universiteit Brussel, Faculty of Economics, Belgium.}
\and
Stephan Sturm \thanks{ %
S. Sturm, Worcester Polytechnic Institute, Department of Mathematical Sciences, 100 Institute Road, Worcester, MA 06109, USA
	(e-mail: \texttt{ssturm@wpi.edu})}
}

\usepackage[authoryear]{natbib}
\begin{document}

\maketitle

\begin{abstract}
This paper studies the topic of cost-efficiency in incomplete markets. A payoff is called cost-efficient if it achieves a given probability distribution at some given investment horizon with a minimum initial budget. Extensive literature exists for the case of a complete financial market. We show how the problem can be extended to incomplete markets and how the main results from the theory of complete markets still hold in adapted form. In particular, we find that in incomplete markets, the optimal portfolio choice for non-decreasing preferences that are diversification-loving (a notion introduced in this paper) must be ``perfectly'' cost-efficient. This notion of perfect cost-efficiency is shown to be equivalent to the fact that the payoff can be rationalized, i.e., it is the solution to an expected utility problem.
\end{abstract}

\vspace{5mm}
 
\begin{flushleft}
\end{flushleft}

\vspace{1cm}

\textbf{Declarations of interest:}   none.
\section{Introduction}

Cost-efficiency of a payoff refers to the property that it achieves a given probability distribution at some given investment horizon with a minimum initial budget. Cost-efficiency has shown to be a key tool for the analysis of several important questions in mathematical finance, e.g., the distribution builder and the quantile formulation of portfolio choices.  The distribution builder \cite{SGB00, S07, goldstein2008choosing} is used to effectively elicit consumer preferences, bypassing utility functions that are theoretically sound, but usually hard to estimate practically. The investor directly selects possible outcomes by determining his desired terminal wealth in one hundred scenarios and receives feedback if the chosen allocation is actually feasible with the available budget. This calculation hinges on a cost-efficiency argument, as it uses the input as probability mass function of the claim and calculates the budget needed for the most cost efficient hedging strategy.

 \cite{carlier2011optimal} as well as \cite{HZ11} provide a formulation of the portfolio selection problem that uses cost efficiency to reduce the problem to an optimization over quantile functions; this approach is in particular useful for settings in behavioural finance. This characterization  dates back to Dybvig's work \cite{D88a,D88b} who first shows in a complete discrete market in which all states are equally probable that cost-efficient payoffs must be non-increasing in the pricing kernel. In particular, solving for the payoff that maximizes a law-invariant increasing objective function then amounts to looking for a non-increasing functional of the  pricing kernel, or, equivalently, for the unknown quantile function (or probability distribution) of the optimum. This is the foundation of the so-called quantile approach to solving a large class of non expected utility maximization problems, e.g., \cite{schied2004neyman}, \cite{carlier2006law, carlier2008two,carlier2011optimal}, \cite{JZ08},  \cite{HZ11} and \cite{BBV14}.

Unfortunately, further progress has been restricted by the fact that the classical cost-efficiency set-up works only in complete markets (\cite{D88a,D88b}, \cite{schied2004neyman}, \cite{carlier2006law, carlier2008two} and \cite{BBV14}). This approach has been criticized as the attainability of the optimal payoff is guaranteed only in a complete market, which might not be the case in an incomplete market. As far as we know, existing results on cost-efficiency in incomplete markets are relatively scarce in the literature. Notable exceptions are the work of \cite{jouini2001efficient} who deals with cost-efficiency in an incomplete discrete market with a finite number of states, and of \cite{HZ11} (in Section 5) for an incomplete market with deterministic interest rate, returns, and volatility processes. These assumptions guarantee in particular that all attainable claims have the same optimal pricing kernel. 

The paper at hand aims to address this issue by providing a theoretical framework for cost-efficiency in incomplete markets and by proving central characterizations of (perfectly) cost-efficient payoffs. The results are based on characterizing the \textit{right} generalizations of the key concepts of the theory of cost-efficiency: Equality in distribution is replaced by a relaxation in convex order, and optimizers that retain the target distribution, termed \textit{perfectly cost-efficient} play the key role in the analysis. Furthermore, preference functionals are required to be \textit{diversification-loving}, a concept to be shown to be sufficient to guarantee perfect cost-efficiency of the optimizer while being weaker than more classical notions as (quasi-)concavity. We also investigate the relationship between cost-efficiency and  expected utility maximization and show that perfectly cost-efficient payoffs are, essentially, optimizers to utility maximization problems.

The paper is organized as follows. In Section~\ref{S2}, we set the playground, lay out the toys, and define the rules of the game. Besides revisiting the general market setup and facts about convex order and convex hulls, we discuss preference functionals in detail. In particular, we introduce the notion of a diversification-loving preference functional that is key for the further analysis and discuss its relation to other, well-known properties of functionals. The standing assumptions for the rest of the paper are spelled out at the end. Section~\ref{S3} extends the notion of cost-efficiency to incomplete markets by relaxing it in convex order (resp.\ convex hulls more generally) and introduces the new notion of perfect cost-efficiency (cost-efficiency while retaining the target distribution) that is crucial for all further results. The first major finding, namely that portfolio optimizers for diversification-loving, increasing and upper semicontinuous preference functionals are perfectly cost-efficient, is given in Theorem~\ref{theo1}.  Section~\ref{relat} provides then an explicit construction of an optimizer to the cost-efficiency problem based on a detailed analysis of related minimax/saddle-point problems. Section~\ref{app} shows that under mild regularity assumptions all perfectly cost-efficient payoffs can be rationalized, i.e., are the optimizers to classical expected utility maximization problems. In fact, Theorem~\ref{thD2} provides a general equivalence statement for many of the notions considered in this paper.  Section~\ref{S7} concludes.

\section{Setting}\label{S2}

We consider a financial market with risky assets modeled as locally bounded $d$-dimensional semimartingale $S$ above an atomless, filtered standard Borel probability space $\offp$ with $\mathbb{F} = (\mathcal{F}_t)_{0 \leq t \leq T}$.\footnote{The results also hold in the case of a discrete probability space with $n$ equiprobable states. This setting is discussed in the companion document  \cite{BS24appendix}.} Further, we assume that the filtration is complete and right-continuous and that $\mathcal{F}_T = \mathcal{F}$. We assume that there exists a riskless asset and, without loss of generality, that it is bearing no interest.

We assume that the market is free of arbitrage in the sense of No Free Lunch with Vanishing Risk (NFLVR). This is equivalent to say that the set of all equivalent martingale measures
\begin{equation*}
	\mathcal{M}^e := \bigl\{ \Q \sim \P \, : \,  S  \text{ is a } \Q\text{-local martingale} \bigr\}
\end{equation*}
is not empty, see Corollary 9.1.2 in \cite{DS06}. Denote the set of the associated pricing kernels by 
\begin{equation*}
	\Xi := \biggl\{\xi = \frac{d\Q}{d\P} \biggr\vert_{\F_T} \, : \, \Q \in \mathcal{M}^e\biggr\}.
\end{equation*}

Moreover, we denote by $\Xibar$ the closure of $\Xi$ in $\lnpofp$ with respect to the Ky-Fan metric (which metrizes convergence in probability),
\[
d_{KF}(X, Y ) = \inf\Bigl\{\varepsilon > 0 \, :\, \mathbb{P}\bigl[\vert X - Y\vert > \varepsilon\bigr] \leq \varepsilon\Bigr\},
\]
and note that both, $\Xi$ and $\Xibar$ are convex sets. 

Consider an investor with initial capital $x_0 \in \R$ and who trades with an admissible, that is predictable, $S$-integrable and self-financing, strategy $H$.  The value process of such a portfolio is then given by
\begin{equation*}
	X_{t}^{x_0,H} := x_0 + \int_{0}^{t} H_{s} \, dS_{s}, \qquad  0\le t\le T.
\end{equation*}
Denote by $\mathcal{X}(x_0)$ the set of all random variables superhedgeable at time $T$  with initial capital $x_0$ and some credit line $C$,
\begin{align}\label{procX}
	\mathcal{X}(x_0) &:= \Bigl\{ Z \in \lnofp \, : \,   Z \leq X_T^{x_0,H}  \mbox{ for some admissible strategy }H, \Bigr. \nonumber \\
     & \phantom{= \Bigl\{ \, } \quad\quad\quad\Bigl. X_t^{x_0,H} \geq -C \text{ for some } C \geq 0 \text{ for all } 0\le t\le T \Bigr\}.
\end{align}
 Moreover let
\begin{equation}\label{XX}
      \mathcal{X} := \bigcup_{x_0 \in \R} \mathcal{X}(x_0)
\end{equation}
be the set of all superhedgeable claims with finite initial budget. Note that 
$\mathcal{X}(x_0) = \mathcal{X}(0) + x_0$ because the credit line $C$ is not a fixed parameter. Indeed if $Z \leq X^{x_0,H}_T$ for some admissible strategy $H$ with $X^{x_0,H}_T \geq -C$ for some $C \geq  0$, then
$Z - x_0 \leq X^{0,H}_T$ and $X^{0,H}_T \geq -(C + x_0)^+$ which proves the claim. We then have \[\mathcal{X} = \mathbb{R} + \mathcal{X}(0).\]
Conversely, we associate to every claim $X \in \mathcal{X}$ its cost $c(X)$, the amount of money needed to superreplicate it in all possible states of the world:
\begin{equation}\label{SHP0}
    c(X) := \inf\{x_{0}\in\mathbb{R} :X\in \mathcal{X}(x_0) \}.
\end{equation}
 It is well-known (see, e.g., Section 3 of \cite{K96} or Th.\,9.5.8 in \cite{DS06}) that
\begin{equation}\label{SHP}
    c(X) = \sup_{\xi \in \Xi} \E[\xi X].
\end{equation}
Therefore, we call it also the superhedging price.

\subsection{Convex Order and Convex Hull of Equally Distributed Variables}

Let $F$ be a distribution (which we identify with its cumulative distribution function, or cdf),  $D(F)$ be the set of all random variables with the distribution $F$ and $\conv(F)$ be the set of convex combinations of random variables with this distribution $F$.  $\convc(F)$ be its closure with respect to the Ky Fan metric. A key concept used throughout the paper is the characterization of the closed convex hull $\convc(F)$ using convex order: Let us denote  $Z \preceq_{cx} X$ if $Z$ is smaller than $X$ in convex order, which means that for all convex functions $\varphi$ such that the expectations exist, $\E[\varphi(Z)]\le \E[\varphi(X)].$ As this property depends only on the distribution, we write similarly $G \preceq_{cx} F$ if $X \preceq_{cx} Y $ for $X \in D(G)$, $Y \in D(F)$ and we allow for mixed use such as $X \preceq_{cx} F$ to denote convex order. Similarly, if $X$ is a random variable with distribution $F$, we use at times also the notation $D(X)$ for $D(F)$. Throughout the paper, we will also use concave order. Recall that $Y$ dominates $X$ in concave order (that we denote by $X\preceq_{cv}  Y$) if and only if $Y\preceq_{cx}  X.$ We recall the following result from Lemma 2.2 in \cite{he2016risk}.

\begin{proposition}\label{cxorder}
	Let $X\in\loofp$ with cumulative distribution function (cdf) $F$, then
	\[
		\convc(F)=\bigl\{ Y \in \loofp \, : \, Y \preceq_{cx} X\bigr\}.
	\]
\end{proposition}
 
\begin{proof}
	Under the assumptions that $(\Omega, \mathcal{F}, \mathbb{P})$ is an atomless standard Borel space, Sections 2.2 \& 2.3 of \cite{he2016risk} imply that the set of random variables smaller than $X$ in convex order is the $L^1$-convex closure of $F$-distributed random variables. Thus
	\[
	\bigl\{ Y \in \loofp \, : \, Y \preceq_{cx} X\bigr\} = \convc^{L^1}(F) \subseteq \convc(F).
	\]
But the $L^1$ closure and the $L^0$ closure actually agree as the set $\bigl\{ Y \in \loofp \, : \, Y \preceq_{cx} X\bigr\}$ is uniformly integrable; therefore, by Vitali's convergence theorem, convergence in probability on this set implies convergence in $L^1$, and the $L^0$- and $L^1$-closures agree. Indeed, the set $D(F)$ is uniformly integrable trivially, and therefore, by de la Vall\'{e}e Poussin's characterization, there exists a convex function $\varphi$, bounded from below and of superlinear growth, e.g., $\varphi(z) = \sum_{n=1}^\infty (z-n)^+$, that $\sup_{Y \in D(F)} \E[\varphi(Y)] <\infty$. But as for $\Lambda^n = \{\lambda^n = (\lambda^n_1, \ldots, \lambda_n^n) \, : \, \lambda_i^n > 0, \, \sum_{i=1}^n \lambda_i^n = 1\}$ we have by Fatou's lemma
\begin{align*}
\sup_{Y \preceq_{cx} X} \E[\varphi(Y)] & =  \sup_{\substack{Y_i^n \in D(F) \\ \lambda^n \in \Lambda^n}}\E\Biggl[\varphi \Biggl( L^1\text{-}\lim_{n \to \infty} \sum_{i=1}^n \lambda_i^n Y_i^n\Biggr)\Biggr] \leq \sup_{\substack{Y_i^n \in D(F) \\ \lambda^n \in \Lambda^n}}\E\Biggl[\varphi \Biggl( \sum_{i=1}^n \lambda_i^n Y_i^n\Biggr)\Biggr] \\ &\leq  \sup_{\substack{Y_i^n \in D(F) \\ \lambda^n \in \Lambda^n}}\E\Biggl[ \sum_{i=1}^n \lambda_i^n \varphi \bigl(Y_i^n\bigr)\Biggr] = \sup_{Y_i^n \in D(F)}\E\bigl[ \varphi \bigl(Y_i^n\bigr)\bigr] < \infty
\end{align*}
it follows, using again de la Vall\'{e}e Poussin, that  $\bigl\{ Y \in \loofp \, : \, Y \preceq_{cx} X\bigr\} $ is uniformly integrable.
\end{proof}
  
\begin{remark}
	Let $F$ be a distribution with finite mean, i.e., $\int_{\mathbb{R}}|x| F(dx)<\infty$. An immediate corollary of Proposition~\ref{cxorder} is that $\convc(F)$  is rearrangement invariant (\cite{delbaen2009}, called law-invariant in Definition 2.1 of \cite{bellini2021}), i.e., if $X$ and $Z$ are r.v. equal in distribution, then $X\in\convc(F)$ implies $Z\in\convc(F).$  
\end{remark}

To appreciate the difference between the general and the integrable setting, let $F$ be a distribution bounded from below. Note that $F$ has finite mean if and only if the convex closure $\convc(F)$ is bounded in probability (see Appendix~\ref{appconvF} for details). On the other hand, the set $D(F)$ is bounded in probability even for $F$ with infinite mean. Thus there is a huge gap between the family $X \in D(F)$ and its convex closure.

\subsection{Preferences}

We consider an investor who seeks  to optimize her wealth according to her preferences given by the functional $V \, : \, \mathcal{D} \to \R \cup \{-\infty\}$, where $\mathcal{D}$ is either $\lnofp$ or a \textit{Schur-convex}\footnote{A set $\mathcal{D} \subseteq \loofp$ is called Schur-convex if $X \in \mathcal{D}$ implies $Y \in \mathcal{D}$ for all $Y \in \loofp$, $Y \preceq_{cx} X$, cf. Definition 3.3 in \cite{bellini2021}. Note in particular that the subspaces $\lpofp$ with $1 \leq p \leq \infty$ are Schur-convex, as are Orlicz spaces $\lphiofp$ with convex Young functional $\Phi.$} subset of $\loofp$ such that  $\mathcal{D}+\mathbb{R}_{\geq 0} \ind_{\Omega} \subseteq \mathcal{D}$. We say that $V$ is \textit{non-decreasing} if, for two payoffs $X$ and $Y$ satisfying $X\leq Y$ almost surely, one has that $V(X)\leq V(Y).$ $V$ is \textit{increasing}, if besides non-decreasing it is strictly increasing under addition of constants, i.e.,  $V(X)<V\bigl(X+\varepsilon \ind_{\Omega}\bigr)$ for every $X \in \mathcal{D}$ and $\varepsilon>0$. We say that $V$ is \textit{law-invariant} (or \textquotedblleft state-independent\textquotedblright) if any two identically distributed payoffs $X$ and $Y$ are valued equally, i.e., $Y \in D(X)$ implies that $V(X)=V(Y)$. 

Recall also that $X\preceq _{fsd}Y$ means that $Y$ is (first-order) stochastically larger than $X$, i.e., for all $x\in \mathbb{R}$, $F_{X}(x)\geq F_{Y}(x),$ where $F_{X}$ and $F_{Y}$ denote the cdfs of $X$ and $Y$ respectively. Equivalently, for all non-decreasing functions $u$ such that both  expectations $\E[u(X)]$ and $\E[u(Y)]$ exist in $\mathbb{R} \cup \{\infty\}$, we have $\E[u(X)]\le \E[u(Y)]$. The functional $V$ is said to respect first order stochastic dominance if $X\preceq _{fsd}Y$ implies that $V(X)\le V(Y)$.  Furthermore, we will say that the functional $V$ is consistent with concave order if $X\preceq_{cv}  Y$ implies that $V(X)\le V(Y)$. 

A preference functional  $V \, : \, \mathcal{D} \to \R \cup \{-\infty\}$ is  non-decreasing and law-invariant if and only if it is consistent with first-order stochastic dominance (see Theorem 1 in \cite{BCV15}). It turns out that a similar characterization can be obtained for concave order if one assumes preference for diversification. The key notion is the following:

\begin{definition}\label{defDiv}
	The preference functional $V$ is said to be diversification-loving if for all cdfs $F$, for all $n\in\mathbb{N}$, for all $X_{i}\in \mathcal{D}$, $i=1,...,n$, such that $X_{i}\in D(F)$, 
	\begin{equation}\label{dive}
		\forall j\in\bigl\{1,2,...,n\bigr\},\quad \forall \lambda_{i}\ge0\ s.t\ \sum_{i=1}^n\lambda_{i}=1, \quad\quad V\biggl(\sum_{i=1}^n\lambda_{i}X_{i}\biggr)\ge V(X_{j}),
	\end{equation}
	and thus by summation,
	\[
		\forall \lambda_{i}\ge0\ s.t\ \sum_{i=1}^n\lambda_{i}=1, \quad\quad V\biggl(\sum_{i=1}^n\lambda_{i}X_{i}\biggr)\ge \sum_{i=1}^n\lambda_{i}V(X_{i}).
	\]
\end{definition}

Observe that if $V$ is diversification-loving then it is law-invariant. To see that, consider $X_{1}$ and $X_{2}$ with cdf $F$, then applying property \eqref{dive} with $\lambda_{1}=1, \lambda_{2}=0$ and $\lambda_{1}=0, \lambda_{2}=1$ we have respectively that $V(X_{1})\geq V(X_{2})$ and $V(X_{2})\geq V(X_{1})$, which proves that $V$ is law-invariant. Furthermore, note that if $V$ is concave and law-invariant, then it is diversification-loving but the reverse implication is not true in general: the ``concavity property'' is not required for arbitrary random variables, but only for those who have the same distribution. It turns out that diversification-loving  is equivalent to consistency with respect to concave order (also called consistent with majorization order or Schur-concave, cf. \cite{Dana2005}). 

Our notion of ``diversification-loving'' is related to, but distinct from, the classical notions of preference for diversification considered by \cite{dekel1989asset} (see Definition 2) and by \cite{chateauneuf2002diversification} (see Definition 1). In those papers, diversification is required for \emph{equally preferred} assets/acts: if $X_1,\dots,X_n$ are all equally preferred, then every convex combination $\sum_i \lambda_i X_i$ is weakly preferred to each $X_j$. In contrast, Definition~\ref{defDiv} only requires this property for assets having the same distribution. Hence, given that law invariance follows from our definition, the classical preference-for-diversification property implies our diversification-loving property, but not the converse. Namely, if claims with equal distribution $F$ are equally preferred, then we can use the results of \cite{chateauneuf2002diversification} and \cite{dekel1989asset} to conclude that the convex combination of these claims is preferred to each of them. However, the reverse implication does not hold true in general. For instance, the functional $V(X):=\E[X]-\Var[X]^{1/4}$ on $L^2$ is diversification-loving, but it does not satisfy the classical preference-for-diversification property: indeed, if $\varepsilon$ is equal to $1$ or $-1$ with probability $\frac{1}{2}$, and one sets $X_1:=6+\varepsilon$, $X_2:=7+4\varepsilon$, and $\lambda_1=\lambda_2=\tfrac12$, then $V(X_1)=V(X_2)=5$ whereas $V\bigl(\frac{X_1+X_2}{2}\bigr)=6.5-\sqrt{2.5}<5$. In fact, an equivalent of our notion with merely two random variables is formulated in Definition 4.7 of \cite{chateauneuf2007sure} and is referred to as ``strong diversification.''

\begin{theorem} \label{prefssd}
	An $L^1$-upper semicontinuous preference functional  $V \, : \, \loofp \supseteq \mathcal{D} \to \R \cup \{-\infty\}$ is diversification-loving  if and only if it is consistent with concave order.
\end{theorem}

\begin{proof}
	First, assume that  $V \, : \, \loofp \supseteq \mathcal{D} \to \R \cup \{-\infty\}$ is  diversification-loving and consider $X$ and $Y$ such that $X\preceq_{cv}Y$. Then, by Proposition~\ref{cxorder}, $Y=\lim_{n\rightarrow+\infty}\sum_{i=1}^n\lambda_i^n X_i^n$ where $X_i^n$ has cdf $F$. Using the diversification-loving property, we have that for all $j$ and $n$
	\[
		V(X_j^n)\le V\biggl(\sum_{i=1}^n \lambda_i^n X_i^n\biggr),
   	\]
	and thus by taking the limit 
	\[
		V(X_j^n)\le \limsup_{n \to \infty} V\biggl(\sum_{i=1}^n \lambda_i^n X_i^n\biggr) \le V\biggl(\lim_{n \to \infty}\sum_{i=1}^n \lambda_i^n X_i^n\biggr) = V(Y)
	\] 
	by upper semicontinuity. By law-invariance of $V$, we have that $V(X_j^n)=V(X)$, thus $V(X)\le V(Y)$. This ends the proof.
	
    For the reverse implication, assume that $V$ is consistent with concave order. Let $X$ and $Y$ have distribution $F$, then $X\preceq_{cv}Y$ and $Y\preceq_{cv}X$ thus $V(X)\leq V(Y)$ and $V(Y)\leq V(X)$. Therefore $V(X)=V(Y)$, i.e., $V$ is law-invariant. Furthermore, let $X_i$ for $i=1,\ldots,n$ be distributed with $F$, and $\lambda_i\in[0,1]$ such that $\sum_i\lambda_i=1$, then for all functions $\varphi$ concave, we have 
    \[
\sum_i\lambda_i \varphi(X_{i})\leq		\varphi\biggl(\sum_i \lambda_i X_{i}\biggr).
	\]
	We then take the expectation on both sides (when these expectations exist) and use the law-invariance property to write that $\E[\varphi(X_{i})]=\E[\varphi(X_{j})]$ for all $i,j$. Thus, for all $\varphi$ concave and all $j\in\bigl\{1,...,n\bigr\}$ such that the following expectations are finite,
	\[
 \E\bigl[\varphi( X_{j})\bigr]  \leq \E\biggl[\varphi\biggl(\sum_i \lambda_i X_{i}\biggr)\biggr]
	\]
	which equivalently means that $X_{j}\preceq _{cv}\sum_i\lambda_i X_{i}$. And as $V$ is consistent with concave order, then $V(X_{j})\leq V(\sum_i \lambda_i X_{i}),$ i.e., $V$ is then diversification-loving.
\end{proof}

Note that Theorem~4.2 in \cite{chateauneuf2007sure} is very much in the spirit of our Theorem~\ref{prefssd}. However, our result is not a direct corollary of theirs, since they work on the space of bounded random variables under compact continuity and strict monotonicity, whereas we consider a  domain $\mathcal D\subseteq L^1(\Omega,\mathcal F,\mathbb P)$, do not assume strict monotonicity or law-invariance, rather derive law-invariance from diversification-loving, and obtain an equivalence with concave order rather than second-order stochastic dominance.

\begin{corollary} \label{prefssd3}
 	For a convex domain $\mathcal{D}$, an $L^1$-upper semicontinuous, law-invariant preference functional  $V \, : \, \loofp \supseteq \mathcal{D} \to \R \cup \{-\infty\}$ that is quasi-concave (i.e., all super-level sets are convex, or, equivalently, $V\bigl(\lambda X +(1-\lambda) Y\bigr) \geq \min\bigl(V(X),V(Y)\bigr)$ for all $X$, $Y \in \mathcal{D}$ and $\lambda \in [0,1]$) is diversification-loving.
\end{corollary}

\begin{proof}
	Fix $X_i \in \mathcal{D}$ with $X_i \in D(F)$. By repeated use of the quasi-concavity, we get that
	\[
		V\biggl( \sum_{i=1}^n \lambda _i X_i\biggr) \geq \min_{i \in \{1, \ldots, n\}} V\bigl(X_i\bigr) = V\bigl(X_j\bigr)
	\]
	for $\lambda_i \geq0$ with $\sum_{i=1}^n \lambda_i = 1$ by law-invariance for any $j \in \{1,\ldots, n\}$. But this means that $V$ is diversification-loving.
\end{proof}

Second-order stochastic dominance  is closely related to concave order defined above. Namely, $X\preceq _{ssd}Y$ means that $Y$ is (second-order) stochastically larger than $X$, i.e., for all $x\in \mathbb{R}  $, $\int_{-\infty}^xF_{X}(t)dt\geq \int_{-\infty}^xF_{Y}(t)dt,$ where $F_{X}$ and $F_{Y}$ denote the cdfs of $X$ and $Y \in \loofp$ respectively. Equivalently, for all concave non-decreasing functions $u$, $\E[u(X)]\le \E[u(Y)]$. We then have the following corollary.

\begin{corollary} \label{prefssd2}
	An $L^1$-upper semicontinuous preference functional  $V \, : \, \loofp \supseteq \mathcal{D} \to \R \cup \{-\infty\}$ is diversification-loving and non-decreasing if and only if it is consistent with second-order stochastic dominance.
\end{corollary}

\begin{proof}
	Assume $V$ is $L^1$-upper semicontinuous, diversification-loving and non-decreasing. Let us first prove that $V$ is consistent with SSD. Let $X$ and $Y$ be two random variables in $\mathcal{D}$ such that $X\preceq_{ssd} Y$. By Theorem 2 of \cite{rothschild1970increasing}, there exists $A \geq 0$ such that  $X \preceq_{cv} (Y - A)$. By Theorem~\ref{prefssd}, $V$ is consistent with concave order, which implies that $V (X) \leq V (Y - A)$. Since $V$ is non-decreasing and $A \geq 0$, we have $V (Y - A) \leq V (Y )$. Hence $V (X) \leq V (Y )$.
	
	The reverse implication is immediate. Assume $V$ is consistent with SSD. To prove it is non-decreasing, assume $X \leq Y$ a.s. Then $X \preceq_{ssd} Y$, so $V (X) \leq V (Y )$. To prove concave-order consistency (hence diversification-loving by Theorem~\ref{prefssd}), assume $X \preceq_{cv} Y$. This implies $X \preceq_{ssd} Y$; therefore $V (X) \leq V (Y )$.
\end{proof}

\subsection{Portfolio problem}
	We will be concerned with the following portfolio optimization problem
	\begin{equation}\label{problem}
    	\sup_{X \in \mathcal{D}_{x_0}} V(X),
	\end{equation}
	maximizing the portfolio value according to the preference functional $V$ given the initial wealth $x_0$ (where the set $\mathcal{D}_{x_0} := \mathcal{X}(x_0) \cap \mathcal{D}$ denotes the set of claims superhedgeable with initial wealth $x_0$, as defined in \eqref{procX}, in the domain of the preference functional). To do so, we will work for the rest of the paper under the following standing assumption:

 \begin{assumption}\label{A2.2}
 	$V$ is an increasing, diversification-loving and upper semicontinuous functional (with respect to convergence in probability in the case $\mathcal{D} = \lnofp$ and with respect to $L^1$-convergence in the case $\mathcal{D} \subseteq \loofp$). Moreover, problem \eqref{problem} has a finite (not necessarily unique) solution $X^*$.
 \end{assumption}

\begin{remark} When $V(X)=\E[U(X)]$ for some utility function $U$, it is clear that $V$ is diversificati\-on-loving (Definition \ref{defDiv}) if  $U$ is concave (corollary of Theorem \ref{prefssd}), and the other properties hold as well. The same is true when $V(X)=-\rho(X)$ for some law-invariant, lower semicontinuous and convex risk measure $\rho$, such as expected shortfall (also  known as average or conditional value at risk). Finally, distortion risk measures, i.e., $\rho_{g}(X) =\int_{0}^\infty g(1-F_{X}(x))dx$, are consistent with concave order when the distortion function $g$ is concave.
\end{remark}

Note that we work with $V$ defined on $\lnofp$ or on some Schur-convex domains $\mathcal{D} \subseteq \loofp$, but that the equivalence of the properties in Assumption~\ref{A2.2} given in Theorem~\ref{prefssd} and Corollary~\ref{prefssd2} is only valid in the latter case, as concave order and second order stochastic dominance are only defined for integrable random variables.

\section{Cost-efficiency and Optimal Portfolio Choice \label{S3}}

In this section, we characterize optimal payoffs solving \eqref{problem} depending on the properties of $V$. In Section \ref{3.1} and \ref{3.2}, we have a slightly broader setting and do not need all assumptions from  Section \ref{S2} and rely only on the assumptions spelled out explicitly.

\subsection{Increasing preferences\label{3.1}}
 Let $V: \mathcal{D}\longrightarrow \mathbb{R} \cup \{-\infty\}$ be a function defined on a nonempty convex subset $\mathcal{D} \subseteq L^{0}(\Omega, \mathcal{F}, \mathbb{P})$ satisfying $\mathcal{D}+\mathbb{R}_{\geq 0} \ind_{\Omega} \subseteq \mathcal{D}$.\footnote{If $X \in \mathcal{D}$ and $a \in \mathbb{R}_{\geq 0}$, then $X+a \ind_{\Omega} \in \mathcal{D}$.} Given $X \in \mathcal{D}$, we let $\operatorname{Upp}_V(X)$ be the upper contour set at $X$  defined by
\begin{equation*}
\operatorname{Upp}_V(X):= \bigl\{\tilde{X} \in \mathcal{D}: V(X) \leqslant V(\tilde{X})\bigr\}. \tag{8}
\end{equation*}

Consider a correspondence $\mathcal{X}: x_{0} \longmapsto \mathcal{X}\bigl(x_{0}\bigr) \subseteq \mathcal{D}$ such that
\begin{itemize}
\item[(a)] $x_{0} \ind_{\Omega} \in \mathcal{X}\bigl(x_{0}\bigr)$;
\item[(b)] If $X \in \mathcal{X}\bigl(x_{0}\bigr)$, then $X+\varepsilon \ind_{\Omega} \in \mathcal{X}\bigl(x_{0}+\varepsilon\bigr)$ for any $\varepsilon>0$.
\end{itemize}
The range of $\mathcal{X}$ is denoted by $\mathcal{X}(\mathbb{R})$. In the context of the financial market, recall that $ \mathcal{X}(x_0)$ defined in \eqref{procX} refers to the set of claims superhedgeable with initial wealth $x_0$. In particular, it satisfies the two properties (a) and (b) outlined above. Recall that  the cost function $c: \mathcal{X}(\mathbb{R}) \rightarrow \mathbb{R} \cup \{-\infty\}$ is defined in \eqref{SHP0} by
\begin{equation}
c(X):=\inf \bigl\{x_{0} \in \mathbb{R}: X \in \mathcal{X}\bigl(x_{0}\bigr)\bigr\} . \label{newcost}
\end{equation}

\begin{theorem} \label{th1}
Fix $x_{0} \in \mathbb{R}$ and let
\begin{equation}\label{Xstardef}
X^{\star} \in \operatorname{argmax}\bigl\{V(X): X \in \mathcal{X}\bigl(x_{0}\bigr)\bigr\} . 
\end{equation}
If $V$ is increasing, then $x_{0}=c\bigl(X^{\star}\bigr)$ and
\begin{equation}
X^{\star} \in \operatorname{argmin}\bigl\{c(X): X \in \operatorname{Upp}_V\bigl(X^{\star}\bigr) \cap \mathcal{X}(\mathbb{R})\bigr\} . \label{eq11}
\end{equation}
\end{theorem}

\begin{proof}
We start by proving that $x_{0}=c\bigl(X^{\star}\bigr)$. Let $y_{0}:=c\bigl(X^{\star}\bigr)$. As $X^{\star}$ is attainable from $x_{0}$ (i.e., $X^{\star} \in \mathcal{X}\bigl(x_{0}\bigr)$), we must have $y_{0} \leqslant x_{0}$. Assume by contradiction that $y_{0}<x_{0}$. There exists $\varepsilon>0$ such that $y_{0}+\varepsilon<x_{0}$. As $y_{0}+\varepsilon$ cannot be a lower bound of $\bigl\{y \in \mathbb{R}: X^{\star} \in \mathcal{X}(y)\bigr\}$, there exists $y \leqslant y_{0}+\varepsilon$ such that $X^{\star} \in \mathcal{X}(y)$. As $x_{0}-y>0$, we can apply property (b) to deduce that $X^{\star}+x_{0}\ind_{\Omega}-y\ind_{\Omega} \in \mathcal{X}\bigl(x_{0}\bigr)$. As $V$ is increasing, we have $V\bigl(X^{\star}+x_{0}\ind_{\Omega}-y\ind_{\Omega}\bigr)>V\bigl(X^{\star}\bigr)$. This contradicts the optimality property in \eqref{Xstardef} of $X^\star$. We have thus proved that $x_{0}=c\bigl(X^{\star}\bigr)$. 

We shall now prove \eqref{eq11}. We obviously have $X^\star\in \operatorname{Upp}_V(X^*) \cap \mathcal{X}(\mathbb{R})$ as $V (X^*) \geq V (X^*)$ and  $X^\star\in \mathcal{X}(x_0)$. Assume by contradiction that there exists $Y\in \operatorname{Upp}_V(X^\star)\cap \mathcal{X}(\mathbb{R})$ such that $c(Y ) < c(X^*) = x_0$. By definition of $c(Y)$, there exists $y\in(c(Y), c(X^\star))$ such that
$Y\in\mathcal{X}(y)$. Moreover, as $Y \in \operatorname{Upp}_V(X^\star)$, we also have $V (Y ) \geq V (X^\star)$. As $y < x_0$, we can
apply property (b) to deduce that $Y + x_0\ind_{\Omega}-y\ind_{\Omega} \in \mathcal{X}(x_0)$. As $V$ is increasing, we have
$V (Y + x_0\ind_{\Omega}-y\ind_{\Omega}) \geq V (Y )$, and we deduce that $V (Y + x_0\ind_{\Omega}-y\ind_{\Omega}) > V (X^\star)$. This contradicts
the optimality property \eqref{Xstardef}. We have thus proved \eqref{eq11}.
\end{proof}

\subsection{Increasing diversification-loving upper semicontinuous  preferences \label{3.2}}

The proof of Theorem \ref{th1} only uses the fact that the preference functionals $V$ are increasing. If we have additional assumptions on $V$, then we can improve the characterization in \eqref{eq11}. Observe indeed that Theorem \ref{th1} implies that
\[
X^{\star} \in \operatorname{argmin}\bigl\{c(X): X \in \mathcal{U}\bigl(X^{\star}\bigr) \cap \mathcal{X}(\mathbb{R})\bigr\}
\]
for any subset $\mathcal{U}\bigl(X^{\star}\bigr) \subseteq \operatorname{Upp}_V\bigl(X^{\star}\bigr)$ containing $X^{\star}$.
We have the following result for the set of increasing, diversification loving and upper semicontinuous functions.
\begin{theorem}\label{theo2early} If $F$ has finite mean and $X^* \in \mathcal{D}$ is an optimizer of \eqref{problem} with cdf $F$ for $V$ increasing, diversification-loving and upper semicontinuous, then $X^*$ also solves the minimization problem:
\begin{equation}\label{cvxminimaxTH31}
\inf \bigl\{c(Z): Z \in \mathcal{X}(\mathbb{R}) \text{ and } Z \preceq_{cx} F\bigr\}.
\end{equation}
	If $F$ does not have a finite mean, the infimum in \eqref{cvxminimaxTH31} must be taken over the convex closure (with respect to the topology generated by convergence in probability) of $F$-distributed random variables by
	\[
		\convc(F) := \overline{\conv \bigl(\bigl\{ X \in \mathcal{X} \, : \, X \in D(F)\bigr\}\bigr)}^{L^0}.
	\]	
\end{theorem}
	Note that as we work on Schur-convex domains (or $\lnofp$) that if $X \in \mathcal{D}$, $X \in D(F)$ it follows that $\convc(F) \subseteq \mathcal{D}$. Moreover, Proposition~\ref{cxorder} has established that the notions agree for distributions with finite mean.
\begin{proof} When $X^{\star} \in L^{1}(\Omega, \mathcal{F}, \mathbb{P})$ we can take $\mathcal{U}\bigl(X^{\star}\bigr)=\bigl\{Z \in \mathcal{X}: Z \preceq_{c x} X^{\star}\bigr\}$. Indeed, from Proposition \ref{cxorder}
\[
\bigl\{Z \in \mathcal{X}: Z \preceq_{c x} X^{\star}\bigr\} \subseteq \overline{\operatorname{conv}}(F)
\]
where $F$ is the distribution of $X^{\star}$. As $V$ is diversification loving, we have $\operatorname{conv}(F) \subseteq \operatorname{Upp}_V\bigl(X^{\star}\bigr)$. As $V$ is upper semicontinuous, we have $\overline{\operatorname{conv}}(F) \subseteq \operatorname{Upp}_V\bigl(X^{\star}\bigr)$, and we deduce the first part of the theorem.

 When $X^{\star}$ is not integrable, we can take $\mathcal{U}\bigl(X^{\star}\bigr)=\overline{\operatorname{conv}}(F)$ and we deduce the second part of the theorem.
\end{proof}

Note that Theorem \ref{theo2early} can be seen as a special case of a more general approach. Let  $\mathcal{V}_{0}$ denote the set  of increasing functions $V: \mathcal{D} \rightarrow$ $\mathbb{R} \cup\{-\infty\}$,  $\mathcal{V}_{1}$ the set of increasing and law-invariant functions, $\mathcal{V}_{2}$ the set of increasing and diversification-loving functions, and $\mathcal{V}_{3}$, the set of increasing, diversification-loving and upper semicontinuous functions. We observe that
\[
\mathcal{V}_{3}\subset\mathcal{V}_{2}\subset\mathcal{V}_{1}\subset\mathcal{V}_{0}
\]
and note that the smaller the set $\mathcal{V}$, the more properties the optimizer in  \eqref{problem} satisfies. Let us define for the set $\mathcal{V}$,
 \[\operatorname{Upp}_{\mathcal{V}}(X):=\bigcap_{V \in \mathcal{V}} \operatorname{Upp}_V (X).\]

In Theorem \ref{theo2early} we prove that $\overline{\operatorname{conv}}({D}(X)) \subseteq \operatorname{Upp}_{\mathcal{V}_{3}}(X)$, but we also observe that $D(X) \subseteq \operatorname{Upp}_{\mathcal{V}_{1}}(X)$ and $\operatorname{conv}({D}(X)) \subseteq \operatorname{Upp}_{\mathcal{V}_{2}}(X)$. Note that the above arguments are straightforward and do not depend on whether markets are complete or incomplete.

\subsection{Cost-efficiency}

In a complete market setting, i.e., where all payoffs are hedgeable,
\begin{align*}
L^0_+(\Omega, \mathcal{F}^S_T, \mathbb{P}) & \subseteq \Bigl\{ Z \in \lnofp \, : \,   Z = X_T^{x_0,H}  \mbox{ for some admissible }H \mbox{ and } x_0 \in \mathbb{R}, \Bigr. \\     & \phantom{= \Bigl\{ \, } \Bigl. X_t^{x_0,H} \geq -C \text{ for some } C \geq 0 \text{ for all } 0\le t\le T \Bigr\}.
\end{align*}
(where $L^0_+(\Omega, \mathcal{F}^S_T, \mathbb{P})$ denotes the cone of positive random variables measurable with respect to the sigma--field generated by the asset price processes up to terminal time), there is a unique pricing kernel $\xi$ so that the amount of money $c(Z)$ needed to replicate a payoff $Z$ paid at time $T$ is given by $\E\bigl[ \xi Z\bigr]$. Then recall that given a cdf $F$ and assuming that $\xi$ is continuously distributed (in a continuous probability space), the solution to 
\begin{equation}
	\min_{Z \in D(F)}\E\bigl[ \xi Z\bigr]\label{CE}
\end{equation}
is almost surely unique and equal to $Z^*:=F^{-1}(1-F_{\xi}(\xi))$ (see \cite{D88a,D88b} for results in a discrete market, and 
\cite{JZ08, HZ11, BBV14} for a general market). If $Z^{\star} \in \mathcal{X}$, the same $Z^{\star}$ also minimizes over $\mathcal{X} \cap D(F)$; otherwise the constrained minimum may be higher. It is the cheapest claim that has distribution $F$ and is thus called \textit{cost-efficient}. Furthermore, under Assumption~\ref{A2.2},  given the distribution of an optimal solution $Z^*$ to \eqref{problem} is $F$, then $Z^*$ is cost-efficient for the distribution $F$, i.e., also solves \eqref{CE}. The proof is straightforward and can be done by contradiction using the law-invariance and increasingness of $V$ (where law-invariance is implied by the diversification-loving property of Assumption~\ref{A2.2}).  Moreover, in this setting of a complete market, the properties of diversification-loving and of upper semicontinuity of $V$ are not needed.

Our objective in this paper is to extend the notion of cost-efficiency to an incomplete market setting and show that the fact that the optimal solution of \eqref{problem} is ``cost-efficient'' in a complete market still holds true in an incomplete market for the appropriate generalization of \eqref{CE}, which is \eqref{cvxminimaxTH31}. We note that even though  the optimization is done over the larger set $\{Z \in \mathcal{X} \, : \, Z \preceq_{cx} F\}$, the optimizer is still attained with distribution $F$ as long as $F$ corresponds to a preference functional satisfying Assumption \ref{A2.2}. In this case the optimization set could actually be replaced by the original $\{Z \in \mathcal{X} \, : \, Z \in D(F)\}$; however, this masks the general structure of the problem as will be made clear in the remainder of this section.
	
From an economic perspective, convex order is a well-established measure for riskiness of distributions, a claim that is larger in convex order is perceived as more risky. In particular any rational agent with a concave utility function will prefer the claim smaller in convex order; and a claim larger in convex order has more mass in the tail of the distribution (see \cite{rothschild1970increasing, rothschild1971increasing} for a detailed discussion).  

This approach is not obvious; the most natural way to extend the above cost-efficiency problem \eqref{CE} to an incomplete market setting is to replace the price using the pricing kernel $\xi$ as in \eqref{CE} by the superhedging price \eqref{SHP} that consists of taking the supremum over all pricing kernels. We would then consider solving the following problem
\begin{equation} \label{dist0}
	\inf_{\substack{Z \in \mathcal{X} \\ Z \in D(F)}} c(Z)=
	\inf_{\substack{Z \in \mathcal{X} \\ Z \in D(F)}} \sup_{\xi \in \Xi} \E\bigl[ \xi Z\bigr],
\end{equation}
where $c(Z)$ denotes the superhedging price \eqref{SHP}. This problem consists, given a distribution $F$ at maturity $T$, in finding the random variable $Z^*$ with distribution $F$ that can be superhedged at lowest cost. This minimal cost is also referred to as the \textit{price of the distribution} $F$. We will show hereafter that this formulation of cost-efficiency makes sense if and only if one can relate it to the solution to the portfolio problem \eqref{problem} introduced previously by means of the distribution $F$ being the distribution of the optimizer (see Corollary~\ref{co3.17} and Theorem~\ref{thD2}).  See also Remark \ref{RK1} hereafter.

Such a relationship with the optimal solution to \eqref{problem} can only be obtained by widening the set  of random variables to those being convex combinations of random variables with distribution $F$ (i.e., to $\conv(F)$) and considering the following cost-efficiency problem
\begin{equation}\label{cvxminimaxA}
	\inf_{Z \in \convc(F)} c(Z)=\inf_{Z \in \convc(F)} \sup_{\xi \in \Xi} \E\bigl[\xi Z\bigr].
\end{equation}
We call a pair $(Z^*, \xi^*)$ a solution to problem \eqref{cvxminimaxA} if $(Z^*, \xi^*) \in \mathcal{Z} \times \mathcal{Y}$ with
\begin{equation}
	\mathcal{Z}  := \argmin_{\substack{Z \in \mathcal{X} \\ Z \in\convc(F)   }} \, \max_{\substack{\xi \in \Xibar}} \E\bigl[ \xi Z\bigr], \quad
	\mathcal{Y}  := \bigcup_{Z \in \mathcal{Z}}\argmax_{\substack{\zeta \in \Xibar}} \E\bigl[ \zeta Z\bigr].\label{7b}
\end{equation}
In a slight abuse of language we also often call  just $Z^*$ a solution to the problem \eqref{cvxminimaxA}.

\begin{remark}
	This definition of a solution is more restrictive than just to require that the value function achieves its optimum. Specifically, in general there exist pairs $(\hat{Z},\hat{\xi})$ with $\hat{Z} \in \mathcal{X}$, $\hat{Z} \in\convc(F)$ and $\hat{\xi} \in \Xibar$ such that
	\[
		\E\bigl[ \hat{\xi} \hat{Z}\bigr] = \inf_{Z \in \convc(F)} \sup_{\xi \in \Xi} \E\bigl[\xi Z\bigr]
	\]
	but nevertheless are not solutions to \eqref{cvxminimaxA}. This stems from the nature of the problem as a saddle-point problem (as opposed to a classical optimization problem). An explicit illustration in a discrete setting for this behavior can be found in \cite{BS24appendix}.
	 Given the character of the problem we will have to always consider solution pairs $(\hat{Z},\hat{\xi})$, even if from an economic perspective we care mainly about the optimal payoff, cf. Section~\ref{relat}.
\end{remark}

The notion of cost-efficiency in an incomplete market presented next will play a central role in the rest of the paper. In the case of distributions with finite mean, the convex hull can be replaced by dominance in using convex order (see Proposition~\ref{cxorder}).

\begin{definition}[Cost-efficiency in an incomplete market]
	We say $X\in\mathcal{X}$ is \textit{cost-efficient} (with associated generalized pricing kernel $\xi$) if it solves the problem \eqref{cvxminimaxA} or, equivalently if it solves \eqref{cvxminimaxTH31} in the case of the cdf $F$ having a finite mean. \label{def:costeff}
\end{definition}

It turns out that the notion of cost-efficiency alone is not sufficient for the discussion of incomplete markets, as the optimizer of the problem \eqref{cvxminimaxA} resp. \eqref{cvxminimaxTH31} does not have to have the distribution $F$. We thus introduce the notion of ``perfect'' cost-efficiency that will be instrumental in the rest of the paper, in particular for the characterization of optimal portfolios.

\begin{definition}[Perfect cost-efficiency in an incomplete market]
	$X\in\mathcal{X}$ is \textit{perfectly cost-efficient} (with associated generalized pricing kernel $\xi$) if it solves problem \eqref{cvxminimaxA} for some cdf $F$ and $X\in D(F)$ (resp. \eqref{cvxminimaxTH31} for some cdf $F$ with finite mean and $X\in D(F)$).
	 \label{def:costeffstrict}
\end{definition}

\begin{remark}\label{rm:cost-oder}
	Assume that $X$ solves \eqref{cvxminimaxA} and has cdf $H$, then $X$  solves \eqref{cvxminimaxA} for all distributions $G$ such that $H\preceq_{cx}G\preceq_{cx} F$.  This is obvious as $\convc(G)\subseteq\convc(F)$ and $X\in\convc(G)$. In particular, if $X$ is cost-efficient for some distribution $F$, it is perfectly cost-efficient for $F_X$, and this is the only distribution for which it is perfectly cost-efficient. The other way around, if $X$ is not perfectly cost-efficient for $F_X$, then it cannot be cost-efficient for any other distribution $F$. 
\end{remark}
	
\begin{remark}\label{RK1}We note that in a complete market with continuously distributed pricing kernel every cost-efficient payoff is automatically perfectly cost-efficient, as the solution to \eqref{dist0} is unique. This is not true in incomplete markets as can be seen, e.g., in the examples  in \cite{BS24appendix}. Only a subset of distributions leads to perfectly cost-efficient payoffs and we will show that these are exactly the distributions that may be solutions to optimal portfolio choices formulated as in  \eqref{problem}. In this context, we also write, by slight abuse of terminology, that the ``distribution'' is perfectly cost-efficient. E.g., in the 3-state market model of \cite{BS24appendix}, we give an explicit example where the set of cost-efficient distributions forms a subspace of dimension 2 in the space of all possible distributions that is of dimension 3.
\end{remark}

\subsection{Optimal Payoffs are Cost-efficient }

Optimal solutions of problem \eqref{problem} can be characterized in a complete market (see Theorem B.1 in \cite{JZ08}, Lemma 1 in \cite{HZ11} and Proposition 4 in\cite{BBV14}). Theorem~\ref{theo1} hereafter extends this characterization to general incomplete markets based on the notion of perfect cost-efficiency defined above.
We are now ready to present the first central result of this paper, namely that preference optimizers are perfectly cost-efficient.

\begin{theorem}\label{theo1}
	Assume that $X^* \in \mathcal{D}$ is an optimizer of the problem \eqref{problem} with cdf $F$. Then $X^*$ is perfectly cost-efficient  (Definition~\ref{def:costeffstrict}).
\end{theorem}

\begin{proof}
	Let $Z^*$ be a minimizer to the optimization problem \eqref{cvxminimaxA} for the cdf $F$. Assume to the contrary that $X^*$ would not be an optimizer of \eqref{cvxminimaxA}. Then we have
	\begin{equation*}
		c(X^*) = \sup_{\xi \in \Xi} \E\bigl[\xi X^*\bigr] >   c(Z^*)
	\end{equation*}
	Therefore we have $Z^* + \bigl( c(X^*) - c(Z^*)\bigr) > Z^*$, while 
	\begin{equation*}
		c\Bigl( Z^* + \bigl( c(X^*) - c(Z^*)\bigr) \Bigr) = c(Z^*) + \bigl( c(X^*) - c(Z^*)\bigr) = c(X^*)
	\end{equation*}
	as the cost functional (superhedging price) $c$ is translation invariant. Thus, as the preference functional $V$ is upper semicontinuous, strictly increasing over constants and  diversification-loving (Definition~\ref{defDiv}), we have by the upper semicontinuity, as $Z^*$ is in the closed convex hull of $F$,
	\begin{equation*}
		V \Bigl( Z^* + \bigl( c(X^*) - c(Z^*)\bigr) \Bigr) > V(Z^*) \geq \limsup_{n \to \infty} V\Bigl( \sum_{i=1}^n \lambda_i^n Z_i^n\Bigr) \geq  \limsup_{n \to \infty} \sum_{i=1}^n \lambda_i^n V\bigl(  Z_i^n\bigr) = V(X^*),
	\end{equation*}
	for $\lambda_i^n > 0$ with $\sum_{i=1}^n \lambda_i^n=1$ for all $n$ and $Z_i^n \in D(F)$  and it follows that the distribution $Z^* + \bigl( c(X^*) - c(Z^*)\bigr)$ can be reached at the same cost as $X^*$ while being strictly preferred, contradicting the optimality of $X^*$. Thus $X^*$ has to be an optimizer of \eqref{cvxminimaxA} for $F$ and $X^*\in D(F)$, i.e., $X^*$ is perfectly cost-efficient.
\end{proof}
	
Theorem~\ref{theo1} gives a characterization of the optimal solution for investors with non-decreasing, diversification-loving and upper semicontinuous preferences. We use this characterization in Section~\ref{app} to show that any optimum for such a preference functional $V$ is also the optimal payoff for an expected utility maximization with a concave utility function.

\begin{corollary}\label{co3.17}
	If $X^*$ is an optimizer of \eqref{problem} with cdf $F$, then the superhedging cost of $X^*$ can simply be computed as infimum over superhedging prices of random variables that have distribution $F$, i.e.,
	\begin{equation} \label{equality}
		\inf_{Z \in \convc(F)} \, \sup_{\xi \in \Xi} \E\bigl[ \xi Z\bigr]= \inf_{Z \in \convc(F)}  \, c(Z)=\inf_{Z \in D(F)}  \, c(Z).
	\end{equation}
\end{corollary}
 
\begin{proof}
	From Theorem~\ref{theo1}, an optimal payoff $X^*$ (for general non-decreasing, upper semicontinuous and diversification-loving preferences $V$) must be perfectly cost-efficient using the formulation \eqref{cvxminimaxA}. Thus $X^*$ is also the solution to the problem $\inf_{Z \in D(F)}  \, c(Z)$.
\end{proof}

\section{Perfectly Cost-efficient Payoffs\label{relat}}

We aim to further characterize perfectly cost-efficient payoffs in terms of their distribution function $F$ and of the pricing kernels of the market. We are guided by the result that in complete markets the cost-efficient payoff with distribution $F$ solving \eqref{CE} can be characterized as $F^{-1}\bigl(1-F_\xi(\xi)\bigr)$ as long as the pricing kernel has a continuous distribution $F_\xi$ (\cite{D88a, BBV14}). Key is to realize that perfect cost efficiency is related to the ability to interchange infimum and supremum in \eqref{cvxminimaxA} and the structure of the solutions to these problems. For all further results we will need the following standing assumption.

\begin{assumption}\label{ass:usc}
We assume that $F$ satisfies $\sup_{\xi \in \Xi} \E[ \xi F^{-1}\bigl(F_\xi(\xi)\bigr)\bigr] < \infty$ and that the mapping $\xi \mapsto \E\bigl[\xi Z\bigr]$ is $L^1$-upper semicontinuous for every $Z \in \mathcal{X} \cap D(F)$.
\end{assumption}
We note that the assumption implies that there exists an $x_0 \in \R$ such that $\E[\xi Z] \leq \E[ \xi F^{-1}\bigl(F_\xi(\xi)\bigr)\bigr] < x_0$ for all $\xi \in \Xi$ and $Z \in \mathcal{X}$. Therefore any $Z \in D(F)$ is superhedgeable at cost $x_0$ and $Z \in \mathcal{X}(x_0) \subseteq \mathcal{X}$. In particular, all distributions $F$ with compact support trivially satisfy the first part of the assumption. The assumption can be weakened by using randomization techniques to split atoms, cf. Section \ref{sec:char}.

The next section will discuss the relations between the different optimization problems (as well as their convexified versions), followed by the characterization of their optimizers.

\subsection{Minimax and Maximin Problems}

The cost-efficiency problem \eqref{dist0} is a \textbf{minimax} problem, as is the problem \eqref{cvxminimaxA} in which we have convexified the original domain $D(F)$. The minimax problem has as counterpart the problem to find the  minimal cost to achieve (at least) this distribution,
\begin{equation}\label{dist-problem}
	\sup_{\xi \in \Xi} \, \inf_{\substack{Z \in \mathcal{X} \\ Z \in D(F)}} \E\bigl[ \xi Z\bigr],
\end{equation}
and referred hereafter as the \textbf{maximin} problem. Similarly to the solution to Problem \eqref{cvxminimaxA},  we call a pair $(Z^*, \xi^*)$ a solution to problem \eqref{dist-problem} if $(Z^*, \xi^*) \in \mathcal{Z}^\prime \times \mathcal{Y}^\prime$ with
\begin{equation}\label{domains}
	\mathcal{Y}^\prime := \argmax_{\xi \in \Xibar} \, \min_{\substack{Z \in \mathcal{X} \\ Z \in D(F)}} \E\bigl[ \xi Z\bigr],\quad\quad
	\mathcal{Z}^\prime := \bigcup_{\zeta \in \mathcal{Y}^\prime}\argmin_{\substack{Z \in \mathcal{X} \\ Z \in D( F)}} \E\bigl[ \zeta Z\bigr],
\end{equation}
where in a slight abuse of language we call sometimes just $Z^*$ a solution to the problem \eqref{dist-problem}. And of course we can consider also the convexified version of this problem by replacing $D(F)$ by $\convc(F)$:
\begin{equation}\label{convexifiedmaximin}
	\sup_{\xi \in \Xi} \, \inf_{\substack{Z \in \mathcal{X} \\ Z \in \convc(F)}} \E\bigl[ \xi Z\bigr].
\end{equation}
Thus we consider in total four different problems, which are related by the following proposition.

\begin{proposition}\label{P2.5}
	For any distribution $F$ it holds that
	\begin{equation}\label{CEFF}
		\sup_{\xi \in \Xi} \, \inf_{\substack{Z \in \mathcal{X} \\ Z \in D(F)}} \E\bigl[ \xi Z\bigr] =  \sup_{\xi \in \Xi} \, \inf_{\substack{Z \in \mathcal{X} \\ Z \in \convc(F)}} \E\bigl[ \xi Z\bigr] =  \inf_{\substack{Z \in \mathcal{X} \\ Z \in \convc(F)}} \, \sup_{\xi \in \Xi} \E\bigl[ \xi Z\bigr] \leq \inf_{\substack{Z \in \mathcal{X} \\ Z \in D(F)}} \, \sup_{\xi \in \Xi}\E\bigl[ \xi Z\bigr] .
	\end{equation}
	Furthermore, if $F$ has finite mean, this is equivalent to
	\begin{equation}\label{CEFF2}
		\sup_{\xi \in \Xi} \, \inf_{\substack{Z \in \mathcal{X} \\ Z \in D(F)}} \E\bigl[ \xi Z\bigr] =  \sup_{\xi \in \Xi} \, \inf_{\substack{Z \in \mathcal{X} \\ F_{Z}\preceq_{cx} F}} \E\bigl[ \xi Z\bigr] =  \inf_{\substack{Z \in \mathcal{X} \\ F_{Z}\preceq_{cx} F}} \, c(Z)\leq \inf_{\substack{Z \in \mathcal{X} \\ Z \in D(F)}} \, c(Z),
	\end{equation}
	where $c(Z)$ denotes the superhedging price \eqref{SHP} of $Z$.
\end{proposition}

\begin{proof}
	The first equality of equation \eqref{CEFF} follows from
	\begin{align*}
 		\sup_{\xi \in \Xi} \, \inf_{\substack{Z \in \mathcal{X} \\Z \in \convc(F)}} \E\bigl[ \xi Z\bigr] & = \sup_{\xi \in \Xi} \, \inf_{\substack{Z_i \in \mathcal{X} \\ Z_i \in D(F)}}  \inf_{n \geq 1 } \inf_{\substack{\sum_{i=1}^n \lambda_i = 1 \\ \lambda_i \geq 0}}  \E\biggl[  \xi \biggl( \sum_{i=1}^n\lambda_i Z_i \biggr)\biggr]  \\ & = \sup_{\xi \in \Xi} \, \inf_{n \geq 1 } \inf_{\substack{\sum_{i=1}^n \lambda_i = 1 \\ \lambda_i \geq 0}} \sum_{i=1}^n\lambda_i \inf_{\substack{Z_i \in \mathcal{X} \\ Z_i \in D(F)}}  \E\bigl[  \xi Z_i\bigr] =  \sup_{\xi \in \Xi}  \inf_{\substack{ Z\in \mathcal{X} \\ Z \in D(F)}}  \E\bigl[  \xi Z\bigr].
	\end{align*}
	For the second equality, we note that $\Xibar$ is bounded in probability (by the definition of boundedness in probability and Markov's inequality) besides being closed and convex. As $\E[ \xi Z]$ is linear and upper semicontinuous in $\xi$ by Assumption \ref{ass:usc}, we can apply the generalization of \v{Z}itkovic's $\lnpofp$-minimax theorem established in Theorem B.3 of \cite{BK17} (expanding on Theorem 4.9 in \cite{Z10}), to conclude that
	\[
		\sup_{\xi \in \Xibar} \, \inf_{\substack{Z \in \mathcal{X} \\ Z \in \convc(F)}} \E\bigl[ \xi Z\bigr] =  \inf_{\substack{Z \in \mathcal{X} \\ Z \in \convc(F)}} \, \sup_{\xi \in \Xibar} \E\bigl[ \xi Z\bigr].
	\]
	Finally, the last inequality in the chain of \eqref{CEFF} follows just from the larger domain of the infimum.   The statement \eqref{CEFF2} follows immediately from \eqref{CEFF} using Proposition~\ref{cxorder} on the characterization of $\convc(F)$.
\end{proof}
	
\begin{corollary}
	An upper bound on the price of a distribution $F$ for which there exists $Z\in\mathcal{X} \cap D(F)$ is given by its expectation,
	\[
		\sup_{\xi \in \Xi} \, \inf_{\substack{Z \in \mathcal{X} \\ Z \in D(F)}} \E\bigl[ \xi Z\bigr] \le \int_\mathbb{R} x \, dF(x).
	\]  
\end{corollary}

\begin{proof}
	From Proposition~\ref{P2.5}, we have that
	\[
		\sup_{\xi \in \Xi} \, \inf_{\substack{Z \in \mathcal{X} \\ Z \in D(F)}} \E\bigl[ \xi Z\bigr] = \sup_{\xi \in \Xi} \, \inf_{Z \in \convc(F)} \E\bigl[ \xi Z\bigr].
	\]
	For a distribution $F$ with finite mean, following Section 2.3 of \cite{he2016risk}, we have that $\int_\mathbb{R} x \, dF(x)$ belongs to $\convc(F)$. Thus
	\[
		\sup_{\xi \in \Xi} \, \inf_{\substack{Z \in \mathcal{X} \\ Z \in D(F)}} \E\bigl[ \xi Z\bigr] \le \sup_{\xi \in \Xi} \E\left[ \xi \int_\mathbb{R} x \, dF(x)\right]= \sup_{\xi \in \Xi} \E[\xi] \, \int_\mathbb{R} x \, dF(x) =  \int_\mathbb{R} x \, dF(x).
	\] 	
	If $F$ has infinite mean, the inequality holds trivially, thus the result follows.
\end{proof}

Equality in \eqref{CEFF2} does not hold in general. However, Corollary~\ref{co3.17} shows that only when the equality \eqref{equality} holds, the distribution $F$ can be the distribution of an optimal solution  to a problem \eqref{problem}. In other words: If $F$ is perfectly cost-efficient, then \eqref{CEFF} holds with equality. In Section~\ref{app}, we will show that this is, in fact, a necessary and sufficient condition, under mild regularity assumptions on $F$  (Theorem~\ref{thD2}). While $\eqref{equality}$ asserts the equality of the value function at optimum of the four problems in \eqref{CEFF}, this does not necessarily imply that the optimizers agree (though the minimax theorem \cite{BK17} asserts that the two middle problems share at least one solution, namely a saddle point, and the perfect cost-efficiency implies that the last two problems have the same optimizers). The relation of the optimizers of the different problems will be discussed in detail in the next section; an insightful illustration in a 3-state market, including a counterexample to equality in \eqref{CEFF2}, is provided in \cite{BS24appendix}. 

\subsection{Characterization of Perfectly Cost-efficient Payoffs}\label{sec:char}

We proceed to analyze the optimizers of the different optimization problems. The strategy is first to characterize the solution of the maximin problem \eqref{dist-problem} and then to show that it agrees with the solution to the actual cost-efficiency problem \eqref{cvxminimaxA}. Alas, while this works fine in the ``nice'' case when $F$ is perfectly cost-efficient, continuous and strictly increasing on its domain and the cdf $F_{\xi^*}$ of the optimal pricing kernel is continuous and strictly increasing as well, the general case is not as straightforward. Specifically, not only non-uniqueness of solutions has to be addressed, but potential atoms in pricing kernels force us to use randomization techniques: For $\xi\in\Xi$, we define the distributional transform of $\xi$  for $x\in\R_{\ge 0}$ and $u\in(0,1)$ as
\begin{equation}\label{DISTTRANS}
	\hat{F}_\xi(x ; u) := \P[\xi < x] + u \P[\xi = x].
\end{equation}
In the special case when $\xi$ is continuously distributed we have $\hat{F}_\xi(x ; u)=F_{\xi}(x)$.

For the construction of the solutions, we have to assure the existence of a randomizing uniform variable $U$ on $\ofp$ independent of $\Xibar$. This might not always be the case, so one has to enlarge the filtered probability space to $(\hat{\Omega}, \hat{\mathcal{F}}^0, \hat{\mathbb{F}}^{0+}, \hat{\mathbb{P}})$ by setting $\hat{\Omega} = \Omega \times [0,1]$ (with elements $\hat{\omega} = (\omega,u)$), $\hat{\mathcal{F}} = \mathcal{F} \otimes \mathcal{B}\bigl([0,1]\bigr)$, $\hat{\mathbb{P}} = \mathbb{P} \otimes \lambda\bigl([0,1]\bigr)$ and $\hat{\mathcal{F}}^0$ the completion of $\hat{\mathcal{F}}$ as well as $\hat{\mathbb{F}}^{0+}$ the smallest right-continuous filtration encompassing $F_t \otimes \mathcal{B}\bigl([0,1]\bigr)$  (with $\mathcal{B}$ denoting the Borel sigma field and $\lambda$ the Lebesgue measure). In this way, $U$ is not only independent of $\xi^*$, but of $\Xibar$. The hedgeable random variables $X \in \mathcal{X}$ correspond then to random variables $\hat{X}(\hat{\omega}) =\hat{X}(\omega, \, \cdot \,) = X(\omega)$, whereas the random variables hedgeable through randomization are $\hat{\mathcal{X}} = \bigl\{\hat{X} \in L^0(\hat{\Omega}, \hat{\mathcal{F}}, \hat{\mathbb{P}}) \, : \, \hat{X}( \, \cdot \, , u) \in \mathcal{X} \text{ for every } u \in [0,1] \bigr\}$. To keep the notation of the paper simple, we will assume that the probability space is rich enough to support the uniform random variables, i.e., the enlargement has been made a priori if necessary. We will denote the solutions as $(Z^*, \xi^*,U) \in \mathcal{Z} \times \mathcal{Y} \times  \mathcal{U}$ resp. $\mathcal{Z}' \times \mathcal{Y}' \times  \mathcal{U}$ where $\mathcal{U}$ denotes the set of all standard uniform random variables independent of $\Xibar$.

In this section, we obtain the following results. First,  Proposition~\ref{prop1} shows that the \textbf{maximin} problem \eqref{dist-problem} has at least one solution  of the form $\bigl(Z^*, \xi^*,U\bigr)$ with $Z^* = F^{-1}(1-\hat{F}_{\xi^*}(\xi^*; U))$, $\xi^* \in \Xibar$ and $U$ is a uniformly distributed variable over $(0,1)$ independent of $\Xibar$. Furthermore, if $\xi^*$ has a continuous distribution and $F$ is continuous, then the solution to the \textbf{maximin} problem \eqref{dist-problem} is unique. Proposition~\ref{CEsolu_prop} then shows that if $(Z^*, \xi^*, U)$ with $Z^* = F^{-1}(1-\hat{F}_{\xi^*}(\xi^*; U))$ is a solution of the \textbf{maximin} problem \eqref{dist-problem} that is perfectly cost-efficient then it is also a solution to the \textbf{minimax} problems \eqref{dist0} and \eqref{cvxminimaxA}. Thus, the \textbf{maximin} problem and the \textbf{minimax} problem share the solution $\bigl(Z^*, \xi^*,U\bigr)$ if $Z^*$ is perfectly cost-efficient. However, the \textbf{minimax} problem might have further solutions $(Z^*, \xi^*, U)$ that cannot be written with $Z^* = F^{-1}(1-\hat{F}_{\xi^*}(\xi^*; U))$, $\xi^* \in \Xibar$. All these properties are clearly present in the simple 3-state model we discuss in \cite{BS24appendix} and we invite the reader to use this model as guidance and reference for the following theoretical results.

\begin{proposition}\label{prop1}
	The problem \eqref{dist-problem} has a solution  $(Z^*, \xi^*) \in \mathcal{Z}^\prime \times \mathcal{Y}^\prime$ (as defined in \eqref{domains}) and $Z^*$ can be expressed as
	\begin{equation}\label{solution}
		Z^* = F^{-1}\bigl(1-\hat{F}_{\xi^*}(\xi^*; U)\bigr)
	\end{equation}
	for some $\xi^* \in \Xibar$ and $U$ a standard uniform random variable independent of $\Xibar$.
\end{proposition}

\begin{proof}
	Consider first the complete market case where the proof consists of studying the solution to \eqref{CE} for a given $\xi\in\Xibar$. Recall that for random variables $\xi$, $Y$ with cdfs $F_\xi$, $F_Y$, and an arbitrary standard uniform distributed random variable $U$ we have the Fr\'{e}chet-Hoeffding bound
	\begin{equation*}
		\E\Bigl[F^{-1}_\xi(U)F^{-1}_Y(1-U)\Bigr] \leq \E[\xi Y].
	\end{equation*}
	If the function $F_\xi$ is continuous and $\E\bigl[F^{-1}_\xi(U)F^{-1}_Y(1-U)\bigr]$ is finite, then the bound is attained if and only if  $(\xi,Y)$ equals $\bigl(F^{-1}_\xi(U),F^{-1}_Y(1-U)\bigr)$ almost surely for some standard uniform $U$. The result is well-known when both $F_\xi$ and $F_Y$ are continuous but holds also when only one of the two distributions is continuous. The proof of the uniqueness is similar to the proof of Corollary 2 of \cite{BBV14} and we recall it here. By Fubini,
	\begin{equation*}
		\E[\xi Y] -\E\bigl[\xi \bigl(F_Y^{-1}(1-F_\xi(\xi)\bigr)\bigr] = \int\limits_{0}^{\infty }\int\limits_{-\infty }^{\infty} \Bigl(\mathbb{P}\bigl[\xi>x,Y>y\bigr]-\mathbb{P} \bigl[\xi>x,F_Y^{-1}\bigl(1- F_\xi(\xi)\bigr)>y\bigr]\Bigr) \, dydx.
	\end{equation*}
	As $F_Y^{-1} \bigl(1-F_\xi(\xi)\bigr)$ has cdf $F_Y$ and is antimonotonic with $\xi$ it minimizes $\E[\xi Y]$ for all $Y$ with cdf $F_Y$. In particular, consider  $Y$ with cdf $F_Y$ that achieves this minimum. Then the left-hand side is equal to 0 whereas by Fr\'echet bounds the function inside the integral is non-negative. Then it must be equal to 0 almost everywhere and for almost every $x$ and $y$
	\begin{equation*}
		\mathbb{P}[\xi>x,Y>y]=\mathbb{P}\Bigl[\xi>x,F_Y^{-1} \bigl(1-F_\xi(\xi)\bigr)>y\Bigr].
	\end{equation*}
	It follows that $(\xi,Y)$ has the same joint cdf as $\bigl(\xi,F_Y^{-1} \bigl(1-F_\xi(\xi)\bigr)\bigr)$, thus they have the same support. Since $F_Y^{-1} \bigl(1-F_\xi(\xi)\bigr)$ is a function of $\xi$, then $Y$ must also be a function of $\xi$ almost surely, and 
	\begin{equation*}
		Y = F_Y^{-1}\bigl(1-F_\xi(\xi)\bigr) \, \text{ a.s.}
	\end{equation*}
	Thus, for given $\xi \in \Xi$, the minimizer of
	\begin{equation}\label{min}
		\inf_{Z \in D(F)} \E[ \xi Z]
	\end{equation}
	is given almost surely by the random variable $F^{-1}\bigl(1-F_\xi(\xi)\bigr)$.

	In the case that $F_\xi$ is not continuous, the lower Fr\'echet-Hoeffding bound can be reached using the  distributional transform $\hat{F}_{\xi}$ defined in \eqref{DISTTRANS}. Then for the fixed standard uniform random variable $U$ independent of $\xi$ it follows that $\hat{F}_{\xi}(\xi,U)$ is standard uniformly distributed and a minimizer of \eqref{min} is given almost surely by the random variable $F^{-1}\bigl(1-\hat{F}_\xi(\xi,U)\bigr)$ (which coincides with $F^{-1}\bigl(1-F_\xi(\xi)\bigr)$ for continuous $F_{\xi}$).
	
	The second step of the proof is to use the above result to find an explicit expression of optimal solutions to the maximin problem \eqref{dist-problem} in an incomplete market. As problem \eqref{dist-problem} has a finite solution, we can pick a sequence $(\xi_n)$, $\xi_n \in \Xi$, such that 
	\begin{equation*}
		\inf_{\substack{Z \in \mathcal{X} \\ Z \in D(F)}} \E\bigl[ \xi_n Z\bigr] 
	\end{equation*}
	converges to the supremum. By assumption, the probability space is rich enough to support a standard uniform random variable $U$ independent of $\Xibar$ for which we have, by Assumption \ref{ass:usc}
	\begin{equation*}
		\E\Bigl[ \xi_n F^{-1}\bigl(1-\hat{F}_{\xi_n}(\xi_n,U)\bigr)\Bigr] =  \inf_{\substack{Z \in \mathcal{X} \\ Z \in D(F)}} \E\bigl[ \xi_n Z\bigr].
	\end{equation*}
	Moreover, as $\Xi$ is bounded in $L^1$, we can extract by Koml\'{o}s's lemma a subsequence $(\xi_{n_k})$ that converges almost surely in Cesar\`{o} sense to some random variable $\xi^* \in \loofp$. As $\zeta_k := \frac{1}{k} \sum_{j=1}^k   \xi_{n_j}$ is in $\Xi$ as this is a convex set, we have $\lim_{k \to \infty} \zeta_k = \xi^* \in \Xibar$  as $\Xibar$ is closed in probability. Finally, as 
	\begin{align*}
\lim_{k \to \infty} \inf_{\substack{Z \in \mathcal{X} \\ Z \in D(F)}} \E\bigl[ \zeta_k Z\bigr]
& = \lim_{k \to \infty}  \E\Bigl[\zeta_k  F^{-1}\bigl(1-\hat{F}_{\zeta_k}(\zeta_k,U)\bigr)\Bigr]  \leq \limsup_{k \to \infty}  \E\Bigl[\zeta_k  F^{-1}\bigl(1-\hat{F}_{\xi^*}(\xi^*,U)\bigr)\Bigr] \\ 
& \leq \E\Bigl[\xi^*  F^{-1}\bigl(1-\hat{F}_{\xi^*}(\xi^*,U)\bigr)\Bigr],
	\end{align*} 
by the Fr\'{e}chet--Hoeffding bounds and the assumption of upper semicontinuity, the supremum is indeed attained in $\xi^*$.
\end{proof}

If  $\xi^* \in \Xibar$ is continuously distributed, then the solution in \eqref{solution} writes simply as
\begin{equation}\label{solutionCONT}
	Z^* =  F^{-1}\bigl(1-F_{\xi^*} (\xi^*)\bigr). 
\end{equation}

\begin{remark}\label{rem:condA}
	Note that it is often hard to show that the optimizer $\xi^*$ has a continuous distribution and this has usually to be done ad hoc. An exception is the following assumption on the market model, which turns out to be sufficient but not necessary: If the set $\Xi$ of pricing kernels is uniformly absolutely continuous with respect to the Lebesgue measure $\lambda$ (i.e., for every $\varepsilon>0$ there exists a $\delta>0$ such that $\lambda[A] < \delta $ for $A \in \mathcal{B}(\R_{\geq 0})$ implies $\P[\xi \in A] < \varepsilon$ for all $\xi \in \Xi$), then all $\xi \in \Xibar$ have a continuous cdf $F_\xi$. The proof of this claim is essentially given in the proof of Proposition 2.6 in \cite{BS14}.
\end{remark}

We will show that the explicit expression of the solution of \eqref{dist-problem} derived in Proposition~\ref{prop1} will also be an explicit solution of the cost-efficiency problem as long as $Z^*$ is perfectly cost-efficient. Thus this optimizer can be actually used to derive hedging strategies that will achieve the cost-efficient payoff. We first prove some auxiliary results before providing the main result in Theorem~\ref{CEsolu}.

\begin{lemma}\label{lem:gen_conc}
	The mapping $\xi \mapsto \E\big[\xi F^{-1}\bigl(1-\hat{F}_\xi(\xi; U)\bigr)\bigr]$ is concave on $\Xibar$ for some r.v. $U$ uniformly distributed over $(0,1)$ and independent of $\Xibar$. If $F$ is continuously distributed, then  the mapping is strictly concave.
\end{lemma}

\begin{proof}
	Let $\xi_1$, $\xi_2 \in \Xibar$ and $\lambda \in (0,1)$. Then
	\begin{align*}
		\E\bigl[\xi_\lambda F^{-1}\bigl(1-\hat{F}_\lambda(\xi_{\lambda};U\bigr)\bigr)\bigr] 
		& =\lambda \E\bigl[ \xi_1F^{-1}\bigl(1-\hat{F}_\lambda(\xi_{\lambda};U\bigr)\bigr] + (1-\lambda)\E\bigl[ \xi_2F^{-1}\bigl(1-\hat{F}_\lambda(\xi_{\lambda};U\bigr)\bigr] \\
		&\geq \lambda \E\bigl[ \xi_1F^{-1}\bigl(1-\hat{F}_{\xi_{1}}( \xi_{1}; U)\bigr)\bigr] + (1-\lambda)\E\bigl[ \xi_2F^{-1}\bigl(1-\hat{F}_{\xi_{2}}( \xi_{2}; U)\bigr)\bigr]
	\end{align*}
	where $\xi_{\lambda}:=\lambda\xi_1 + (1-\lambda) \xi_2$, $\hat{F}_\lambda$ is its distributional transform, and $U$ is a uniform over $(0,1)$ independent of $\Xibar$. The inequality follows from the Fr\'{e}chet--Hoeffding bounds. If $F$ is continuous then the strict concavity follows.
\end{proof}

\begin{corollary}\label{cor:maxminminmax}
	If $(Z^*, \xi^*, U)$ with $Z^* \in D(F)$, $Z^* = F^{-1}(1-\hat{F}_{\xi^*}(\xi^*; U))$, $\xi^* \in \Xibar$ and $U$ independent of $\Xibar$, is a solution to the (convexified) minimax problem \eqref{dist0} (resp. \eqref{cvxminimaxA}), then it is also a solution to the maximin problem \eqref{dist-problem} and, a fortiori, the convexified maximin problem \eqref{convexifiedmaximin}. If such $Z^*$ solves any of the four equivalent problems in \eqref{CEFF} for a perfectly cost-efficient cdf  $F$, then, if $F$ is continuous, it is the unique solution of this type, up to the choice of the randomizer $U$. In particular, it is the unique solution of this type if $\xi^*$ is continuously distributed.
\end{corollary}

\begin{proof}
	Given $Z^*$, $\xi^*$ has to be a maximizer of \eqref{dist-problem} over the set of pricing kernels by the equality of the value function of the four problems. Moreover, if there are two distinct optima $(Z_1^*,\xi_1^*)$ and $(Z_2^*,\xi_2^*)$ for any problem, they are distinct optima for the convexified maximin problem \eqref{convexifiedmaximin}. But strict concavity in Lemma~\ref{lem:gen_conc} would allow to construct $\xi_{\lambda} =\lambda\xi_1 + (1-\lambda) \xi_2$, $\lambda \in (0,1)$, with
	\begin{align*}
		\E\bigl[\xi_\lambda F^{-1}\bigl(1-\hat{F}_\lambda(\xi_{\lambda};U\bigr)\bigr)\bigr]> \max{\Bigl(\E\bigl[ \xi_1F^{-1}\bigl(1-\hat{F}_{\xi_{1}}( \xi_{1}; U)\bigr)\bigr],\E\bigl[ \xi_2F^{-1}\bigl(1-\hat{F}_{\xi_{2}}( \xi_{2}; U)\bigr)\bigr]\Bigr)}
	\end{align*}
violating the optimality of $\xi_1^*$ and $\xi_2^*$ for this problem.
\end{proof}

The following proposition provides the converse result, namely that the solution to the maximin problem \eqref{dist-problem} is also a solution to the minimax problem \eqref{dist0} as long as it is perfectly cost-efficient.

\begin{proposition}\label{CEsolu_prop}
	Let $(Z^*, \xi^*, U)$ with $Z^* \in D(F)$, $Z^* = F^{-1}(1-\hat{F}_{\xi^*}(\xi^*; U))$, $\xi^* \in \Xibar$ and $U$ independent of $\Xibar$, be a perfectly cost-efficient solution to the maximin problem \eqref{dist-problem} for the cdf $F$. Then it is also a solution to the minimax and convexified minimax	 problems \eqref{dist0} and \eqref{cvxminimaxA}.
\end{proposition}

\begin{proof}
	We note that if 
	\begin{equation}\label{eq:supequal}
		\sup_{\zeta \in \Xi} \E\bigl[\zeta F^{-1}\bigl(1-\hat{F}_{\xi^*}(\xi^*; U)\bigr)\bigr] = \E\bigl[\xi^* F^{-1}\bigl(1-\hat{F}_{\xi^*}(\xi^*; U)\bigr)\bigr],
	\end{equation}
	 it follows that  $(Z^*, \xi^*)$ is a solution to the (convexified) minimax problem as $Z^* = F^{-1}\bigl(1-\hat{F}_{\xi^*}(\xi^*; U)\bigr)$ achieves the infimum in
	\[
		\inf_{Z \in D(F)} \sup_{\xi \in \Xi} \E[\xi Z] \qquad \Bigl(\text{resp. } \inf_{Z \in \convc(F)} \sup_{\xi \in \Xi} \E[\xi Z] \Bigr).
	\]

	The proof of \eqref{eq:supequal} relies on a fixed-point type argument using the convexly compact version of the Knaster--Kuratowski--Mazurkiewicz (KKM) lemma from Section 4.2 of \cite{Z10}: For $\xi \in \Xibar$ define
	\begin{align*}
		\mathcal{B} (\xi) := \Bigl\{ \zeta \in \Xibar \, : \, & \E\bigl[\xi F^{-1}\bigl(1-\hat{F}_{\zeta}(\zeta; U)\bigr)\bigr] \leq \E\bigl[\zeta F^{-1}\bigl(1-\hat{F}_{\zeta}(\zeta; U)\bigr)\bigr]\Bigr\},
	\end{align*}
	and $\mathcal{A}(\xi)$ as its convex closure
	\[
		\mathcal{A} (\xi) := \overline{\conv{\bigl(\mathcal{B} (\xi)\bigr)}}^{L^0}.
	\]
	We want to prove that the family $\bigl(\mathcal{A} (\xi)\bigr)_{\xi \in \Xibar}$ has the KKM property (see Definition 4.6 of \cite{Z10}), i.e., each $\mathcal{A} (\xi)$ is closed and for each finite set $\{ \xi_1, \ldots , \xi_n\} \subseteq \Xibar$ we have
	\[
		\conv\bigl(\{ \xi_1, \ldots , \xi_n\}\bigr) \subseteq \bigcup_{i=1}^n \mathcal{A}(\xi_i).
	\]
	To show that, note that for arbitrary $\tilde{\xi} \in \conv \bigl(\{ \xi_1, \ldots , \xi_n\}\bigr)$ there exist $\lambda_i^n \geq 0$, $\sum_{i=1}^n \lambda_i^n =1$ with $\tilde{\xi} = \sum_{i=1}^n \lambda_i^n \xi_i$. We have thus 
	\begin{align*}
		\E\bigl[\tilde{\xi} F^{-1}\bigl(1-\hat{F}_{\tilde{\xi}}(\tilde{\xi}; U)\bigr)\bigr] & =  \sum_{i=1}^n \lambda_i^n \E\bigl[\xi_i F^{-1}\bigl(1-\hat{F}_{\tilde{\xi}}(\tilde{\xi}; U)\bigr)\bigr] \\ & \geq \min_{i \in \{1,\ldots, n\}} \E\bigl[\xi_i F^{-1}\bigl(1-\hat{F}_{\tilde{\xi}}(\tilde{\xi}; U)\bigr)\bigr]  = \E\bigl[\xi^\wedge F^{-1}\bigl(1-\hat{F}_{\tilde{\xi}}(\tilde{\xi}; U)\bigr)\bigr]
	\end{align*}
	where $\xi^\wedge$ denotes a minimizer; we have proved that $\tilde{\xi} \in \mathcal{A}(\xi^\wedge) \subseteq \bigcup_{i=1}^n \mathcal{A}(\xi_i)$ concluding the proof of the KKM property.

	Thus, we can apply the KKM lemma, Theorem 4.8 of \cite{Z10}: As all sets $\mathcal{A}(\xi)$ are closed and convex subsets of $\Xibar$, they are all convexly compact by Theorem 3.1 of \cite{Z10}. Thus we can conclude by the KKM lemma that
	\[
		\bigcap_{\xi \in \Xibar} \mathcal{A}(\xi) \neq \emptyset.
	\]
	In other words, there is $\hat{\zeta} \in \Xibar$ such that for each $\xi \in \Xibar$ there exists a sequence $\zeta_i^n \in \Xibar$ and $\lambda_i^n \geq 0$, $\sum_{i=1}^n \lambda_i^n =1$ such that
	\begin{equation} \label{eq:universal}
		\E\bigl[\xi F^{-1}\bigl(1-\hat{F}_{\zeta_i^n}(\zeta_i^n; U)\bigr)\bigr] \leq \E\bigl[\zeta_i^n F^{-1}\bigl(1-\hat{F}_{\zeta_i^n}(\zeta_i^n; U)\bigr)\bigr],
	\end{equation}
	and $\sum_{i=1}^n \lambda_i^n \zeta_i^n \to \hat{\zeta}$. Setting $\xi = \xi^*$, we have by the Fr\'{e}chet--Hoeffding bounds
	\[
		\E\bigl[\xi^* F^{-1}\bigl(1-\hat{F}_{\xi^*}(\xi^*; U)\bigr)\bigr] \leq \E\bigl[\xi^* F^{-1}\bigl(1-\hat{F}_{\zeta_i^n}(\zeta_i^n; U)\bigr)\bigr] \leq \E\bigl[\zeta_i^n F^{-1}\bigl(1-\hat{F}_{\zeta_i^n}(\zeta_i^n; U)\bigr)\bigr];
	\]
	combining this result with Lemma~\ref{lem:gen_conc} yields
	\[
		\E\bigl[\xi^* F^{-1}\bigl(1-\hat{F}_{\xi^*}(\xi^*; U)\bigr)\bigr] = \E\bigl[\zeta_i^n F^{-1}\bigl(1-\hat{F}_{\zeta_i^n}(\zeta_i^n; U)\bigr)\bigr].
	\]
	Again, by Lemma~\ref{lem:gen_conc} we have therefore 
	\[
		\zeta_i^n \in \mathcal{C} := \argmax_{\xi \in \Xibar} \E\bigl[\xi F^{-1}\bigl(1-\hat{F}_\xi(\xi; U)\bigr)\bigr] 
	\]
	for all $\zeta_i^n$. Noting that $\mathcal{C}$ is convex and $\xi^* \in \mathcal{C}$ we conclude that \eqref{eq:universal} implies
	\begin{equation}
		\sup_{\xi \in \Xi} \E\bigl[\xi F^{-1}\bigl(1-\hat{F}_{\xi^*}(\xi^*; U)\bigr)\bigr] \leq \E\bigl[\xi^* F^{-1}\bigl(1-\hat{F}_{\xi^*}(\xi^*; U)\bigr)\bigr]
	\end{equation}
	and thus \eqref{eq:supequal}.
\end{proof}

\begin{corollary}\label{CO4.10}
	If $(Z^*, \xi^*)$ with $Z^* \in D(F)$, $Z^* = F^{-1}(1-F_{\xi^*}(\xi^*))$, for some continuously distributed $\xi^*$, is a solution to the maximin problem \eqref{dist-problem} for a perfectly cost-efficient cdf $F$, then if $F$ is continuous, it is also the unique solution to the minimax (and convexified minimax) problems \eqref{dist0} (and \eqref{cvxminimaxA}).
\end{corollary}

More specifically,  the first part of the solution pair $(Z^*, \xi^*)$ is for continuously distributed $\xi^*$ always unique, but in the case of a discontinuous cdf $F$ there might be more than one (generalized) pricing kernel reaching the optimum.

We can combine Proposition \ref{CEsolu_prop} with Proposition \ref{prop1} to conclude.

\begin{theorem}\label{CEsolu}
	Let $F$ be a distribution such that there exists a perfectly cost-efficient $Z\in\mathcal{X} \cap D(F)$. Then the problem \eqref{cvxminimaxA} has a solution  $(Z^*, \xi^*, U) \in \mathcal{Z} \times \mathcal{Y} \times  \mathcal{U}$ (as defined in \eqref{7b}), where $Z^*$ can be expressed as
	\begin{equation}\label{counter}
		Z^* = F^{-1}(1-\hat{F}_{\xi^*}(\xi^*; U))
	\end{equation}
	for $\xi^* \in \Xibar$ and $U$ standard uniform and independent of $\Xibar$.
\end{theorem}

This theorem is directly related to Theorem 3 on page 354 of \cite{jouini2001efficient} in which the probability space $\Omega$ is finite   and where  \eqref{counter} is written in more explicit form by stating that the claim values (here $Z^*$) are in reverse order compared to the pricing kernel (here $\xi^*$).

\begin{remark}\label{optimizer_char}
	The generalized pricing kernels $\xi^*$ in the optimizers of any of the problems are not necessarily true pricing kernels: neither are they necessarily almost surely strictly positive, nor is their expectation necessarily one, only $\E[\xi^*] \leq 1$ follows from Fatou's lemma.  
\end{remark}

\begin{example} \label{coexa0}
	The optimum to \eqref{dist-problem} is not always a function of any pricing kernel as can be seen in the following example. Consider a Black--Scholes model in which $\mu=r=0$ then the only pricing kernel $\xi$ is equal to 1 a.s. Let $F_{S}$ be the cdf of the risky asset at time $T$. It is continuously distributed (lognormal). Let $(Z^*,\xi^*)$ be a solution to \eqref{dist-problem}. Then the expression \eqref{solutionCONT} cannot be an optimum as $\xi^*=\xi$ does not have a continuous distribution.
\end{example}

The following proposition links our findings to the notion of attainability, i.e., the existence of a hedging strategy that perfectly replicates the payoff.

\begin{proposition}\label{attn}
	Consider a payoff $X^* \in \mathcal{X}$ distributed with a continuous cdf $F$. If $X^*$ is attainable and can be written as $F^{-1}\bigl(1-\hat{F}_{\xi}(\xi; U)\bigr)$ for some $\xi \in \Xibar$ and $U$ standard uniformly distributed and independent of $\Xibar$, then it is perfectly cost-efficient.
\end{proposition}

\begin{proof}
	As $X^*$ is attainable, we know that $\E\bigl[\zeta X^*\bigr]$ is constant and independent of $\zeta \in \Xi.$ Thus
	\[
		\sup_{\zeta \in \Xi} \E\bigl[\zeta X^*] = \E\bigl[\xi X^*] = \E\bigl[\xi F^{-1}\bigl(1-\hat{F}_{\xi}(\xi; U)\bigr)] = \inf_{\substack{Z \in \mathcal{X} \\ Z \in D(F)}} \E\bigl[\xi Z] \leq \sup_{\zeta \in \Xi} \inf_{\substack{Z \in \mathcal{X} \\ Z \in D(F)}} \E\bigl[\zeta Z],
	\]
	so equality holds in \eqref{CEFF}. Proposition~\ref{P2.5} implies that there exists $\xi^* \in \Xibar$, $U$ independent of $\Xibar$ and standard uniform that $F^{-1}(1-\hat{F}_{\xi^*}(\xi^*; U))$ is a solution to the maximin problem \eqref{dist-problem}. Proposition~\ref{CEsolu_prop} in turn implies that it solves also the minimax problem. Thus $X^*$ is a solution to $\inf_{Z \in \mathcal{X}, \, Z \in D(F)} c(Z)$, i.e., perfectly cost-efficient.
\end{proof}

\section{Rationalizing Preferences in Incomplete Markets \label{app}}

On the one hand, we note that it is usually very difficult to determine in a non-parametric way the risk aversion of an investor and the specific utility function that she aims to optimize. On the other hand, \cite{SGB00, S07, goldstein2008choosing} show that it is possible to  ask investors to choose a desired distribution for their final wealth. This approach, also referred as the ``distribution builder'' approach, works well in a complete market as all distributions of wealth can be attained in a cost-efficient way, and may be obtained by maximizing some expected (concave) utility of terminal wealth. It thus reconciles the approach of \cite{goldstein2008choosing} with the expected utility theory by providing the specific concave  utility function that needs to be optimized to find out an optimal wealth process that has the desired distribution at maturity  (see \cite{BCV15} for its explicit expression in a complete market). However, the ``distribution builder'' approach faces limitations in incomplete markets if the ``desired target'' distribution is not attainable or can be proven to be not optimal for any rational investor. In this section, we extend these results in an incomplete market. This can be seen as an inverse of Theorem~\ref{theo1}: not only are preference optimizers perfectly cost-efficient, but (under slight regularity assumptions on the distributions involved), perfectly cost-efficient solutions are optimizers to some preference maximization problem. This, as well as some further equivalent formulations will be presented in Theorem~\ref{thD2}.

Specifically, we consider the following large class of utility functions in line with assumptions made in the literature (e.g., \cite{KS99,DS06}). We assume that ${u}:{\mathbb{R}_{> 0}}\rightarrow \R$ is increasing, continuously differentiable and strictly concave. We extend the utility function by setting $u(0) = \lim_{x \searrow 0} u(x)$ and $u(x) = -\infty$ for $x < 0$. We require the  so-called Inada conditions, i.e., the marginal utility tends to zero when wealth tends to infinity and  to infinity when wealth tends to zero, $\lim_{x\nearrow +\infty}u^{\prime }(x)=0$ and  $\lim_{x\searrow 0}u^{\prime}(x)=+\infty$. In addition, the asymptotic elasticity is assumed to be strictly less than 1, i.e., $\limsup_{x\rightarrow+\infty}\frac{xu^\prime(x)}{u(x)}<1.$   Examples of such utility functions are $u(x)=\ln(x)$ or $u(x)=x^p/p$ for $p<1$, $p\neq0$ and $x>0$.

\begin{definition}[Rationalization by Expected Utility Theory]\label{rational}
	A portfolio choice $X^{\ast} \in D$ with a finite budget $x_0$ is \emph{rationalizable} by expected utility theory if there exists a utility function $ u$ as described above such that $X^{\ast}$ is also the optimal solution to 
	\begin{equation}\label{eu2b}
		\sup_{X\in\mathcal{D}_{x_0}}\E\bigl[ {u}(X)\bigr].
	\end{equation}
\end{definition}

\subsection{Link between Expected Utility Theory and Perfect Cost-Efficiency}

To set the stage we remind the reader of the results on the duality approach to portfolio optimization as given by \cite{KS99} on the optimal solution to the expected utility maximization in general incomplete markets.

\begin{lemma}[Duality Theorem for Expected Utility Maximization]\label{lem4.2}
	Consider a utility function ${u}$ as introduced above. The optimal solution $X^{\ast }$ to the expected
	utility maximization \eqref{eu2b} exists, is a.s. unique and is given by 
	\begin{equation*}
		X^{\ast }:=\bigl( {u}^{\prime }\bigr)^{-1}\bigl( \lambda^{\diamond}\xi^\diamond\bigr),
	\end{equation*}%
	where the constant $\lambda ^{\diamond}>0$ is chosen such that $\E[\xi^\diamond X^{\ast }] =x_{0}$, and where $\xi^\diamond$ is the optimal generalized pricing kernel in
	\[
		\inf_{y >0}\inf_{\xi \in \Xibar} \Bigl(\E\bigl[\widetilde{u}(y\xi)\bigr] + yx_0\Bigr),
	\]
	for $\widetilde{u}$ the Fenchel--Legendre transform of $u$ defined by 
	\[
		\widetilde{u}(y)=\sup_{x\in\mathbb{R}} \bigl(u(x)-xy\bigr).
	\]
\end{lemma}

\begin{proof}
	See Theorem 2.2 in \cite{KS99} or Theorem 3.2 in \cite{BTZ04}.
\end{proof}

Using these results, we can now derive the main equivalence result.

\begin{theorem}	\label{thD2}
	Consider a terminal consumption $X^*$ purchased with an initial budget $x_{0}$ and distributed with a continuous cdf $F$ strictly increasing on $(0,\infty)$ and satisfying $F(x) =  1-F_{\xi^*} (\ell(x)/x^p)$  on some neighborhood $[\alpha, \infty)$ of infinity for some  $p < 1$ and $\ell$ slowly varying. Let $\xi^*$ be a generalized pricing kernel with continuous cdf $F_{\xi^*}$ strictly increasing on its support. Then the following conditions are equivalent:
	\begin{itemize}
		\item[(i)] $X^*$ is a solution to a portfolio maximization problem \eqref{problem} for some $V \, : \, \mathcal{D} \to \R \cup \{-\infty\}$ that is increasing, diversification-loving and upper semicontinuous.
	
		\item[(ii)] $X^*$ is perfectly cost-efficient with cdf $F$ and associated pricing kernel $\xi^*$ (i.e., $(X^*,\xi^*)$ is solving problem \eqref{cvxminimaxA}, Definition~\ref{def:costeffstrict}).
			
		\item[(iii)] $X^*$ is rationalizable by Expected Utility Theory (as per Definition~\ref{rational}).
			
		\item[(iv)] There is equality in \eqref{CEFF} for the cdf $F$, i.e.,
		\begin{equation}\label{EQUA}
			\sup_{\xi \in \Xi} \, \inf_{\substack{Z \in \mathcal{X} \\ Z \in D(F)}} \E\bigl[ \xi Z\bigr] =  \inf_{\substack{Z \in \mathcal{X} \\ Z \in D(F)}} \, \sup_{\xi \in \Xi}\E\bigl[ \xi Z\bigr],
		\end{equation}
		and $(X^*,\xi^*)$ is an optimizer to both problems.
			
		\item[(v)] $X^*$ with cdf $F$ is attainable and can be written as $F^{-1}(1-{F}_{\xi}(\xi))$ for some $\xi\in\Xibar$.
	\end{itemize}
\end{theorem}

\begin{proof}
	The proof follows from these observations:

	$(i)\Rightarrow(ii)$ by Theorem~\ref{theo1}. 
		
	$(ii)\Rightarrow (iv):$ If $X^*$ is perfectly cost-efficient, it solves \eqref{cvxminimaxA} by definition and has cdf $F$ and thus it also solves
	\[
		\inf_{\substack{Z \in \mathcal{X} \\ Z \in D(F)}} \, \sup_{\xi \in \Xi}\E\bigl[ \xi Z\bigr],
	\]
	therefore there is equality  in \eqref{CEFF}. Moreover, Corollary~\ref{cor:maxminminmax} asserts that it is also a solution to the maximin problem. 
		
	$(iv)\Rightarrow (ii):$ If there is equality  in \eqref{CEFF} and $X^*$ is defined as the solution to $\inf_{\substack{Z \in \mathcal{X} \\ Z \in D(F)}} c(Z),$ then it follows that it is perfectly cost-efficient. 
		
	$(ii)\Rightarrow (iii):$ We define the utility function
	\[
		{u}_{\xi^*}(x):= \int_{c}^{x}F_{\xi^*}^{-1}\bigl(1-F(y)\bigr) \, dy
	\]
	for some constant $c\in\R \cup \{+\infty\}$ and check directly (as $F$ and $F_{\xi^*}$ are continuous and strictly increasing) that this is a strictly increasing, strictly concave, continuously differentiable function that satisfies the Inada conditions.
		
	By Karamata's theorem, the asymptotic elasticity condition follows from that $F(x) = 1-F_{\xi^*} (\ell(x)/x^p)$ with $p<1$ and $\ell$ slowly varying. For $p \in [0,1)$ and $c \in \mathbb{R}$, this follows from Thm 1.5.11(i) in \cite{bingham1989regular} (with $\sigma=0$ and $\rho \in [-1,0)$), for $p <0$ and $c = + \infty$ from Thm 1.5.11(ii) in \cite{bingham1989regular} (with $\sigma=0$ and $\rho < -1$). This is actually close to a sufficient condition, as the inverse is implied (except for the case $p=0$) by Thm 1.6.1 of \cite{bingham1989regular} if the defining inequality in the definition of asymptotic elasticity holds for the limit (instead of the limit superior).

	Thus, we can apply Lemma~\ref{lem4.2}, the optimal solution to 
	\begin{equation}
		\sup_{X\in\mathcal{D}_{x_0}}\E\bigl[ {u_{\xi^*}}(X)\bigr]
	\end{equation}
	in which $x_{0}$ is the superhedging price of $F^{-1}\bigl(1-F_{\xi^*}(\xi^*)\bigr),$ writes as $\bigl( {u}_{\xi^*}^{\prime }\bigr)^{-1} \bigl( \lambda^{\diamond}\xi^\diamond\bigr)$. We then replace $\bigl( {u}_{\xi^*}^{\prime }\bigr)^{-1}$ by its explicit expression and we find that
	\[
		X^*=F^{-1}\bigl(1-F_{\xi^*}(\lambda^\diamond\xi^\diamond)\bigr).
	\]
	But $\lambda^\diamond=1$ because of the specific choice of $x_{0},$ and of the fact that $X^*$ is attainable; indeed, if it were not, there would be an attainable random variable $Z \in \mathcal{X}(x_0)$, $Z \geq X^*$, $\P[Z>X^*]>0$ contradicting the optimality of $X^*$ in the utility maximization problem, as $u_{\xi^*}$ is strictly increasing. As $F^{-1}\bigl(1-F_{\xi^*}(\cdot )\bigr)$ is strictly increasing (and thus injective), we have to have $\xi^\diamond = \xi^*$ a.s.
				
	$(iii)\Rightarrow (i):$ By choosing the preference functional $V(X)=\E[u(X)]$ one gets a function that is increasing and diversification-loving. In the case $\mathcal{D} = \lnofp$, $L^0$-upper semicontinuity is implied directly from Fatou's lemma as $u$ is necessarily bounded, implying that convergence in probability of a sequence $(X_n)$ implies $L^1$-convergence of the uniformly bounded sequence $u(X_n$). If $\mathcal{D} \subseteq \loofp$, we can decompose $u = u^+ - u^-$ and by the same Fatou argument $\E[u^-(\, \cdot \, )]$ is lower semicontinuous. $\E[u^+(\, \cdot \, )]$ is actually a continuous functional by Vitali's convergence theorem; the family $\{u^+(X) \, : \, X \in \mathcal{D} \subseteq \loofp\}$ is uniformly integrable by de la Vall\'{e}e Poussin's criterion with $\bigl(u^+\bigr)^{-1}$ acting as Young function. Indeed, the asymptotic elasticity condition implies by Lemma 6.5 of \cite{KS99} that
	\[
		\lim_{x \to \infty} \frac{u^+(x)}{x} = 0 \qquad \text{whence} \qquad \lim_{x \to \infty} \frac{\bigl(u^+\bigr)^{-1}(x)}{x} = +\infty.
	\]
		
	$(v)\Rightarrow (ii)$ by Proposition~\ref{attn}.
		
	$(iii) \, \& \, (iv) \Rightarrow (v):$  Assume $X^* \in \mathcal{X}(x_0)$ were not attainable. Then, by the definition of $\mathcal{X}(x_0),$  there exists an attainable $Y \in \mathcal{X}(x_0)$ such that $Y \geq X^*$ with $\P[Y>X^*]>0$. However, as per (iii) $X^*$ is rationalizable, there exists a strictly increasing utility function $u$ for which $X^*$ is the expected utility maximizer, hence $\E[u(Y)] > \E[u(X^*)]$, contradicting the optimality of $X^*$. Thus, $X^*$ has to be attainable.
		
	As $X^*$ is a solution to the problem \eqref{dist-problem} for some $\xi^* \in \Xibar$, we conclude by Proposition~\ref{prop1} that $X^* = F^{-1}\bigl(1-\hat{F}_{\xi^*}(\xi^*; U)\bigr)$, which simplifies as $\xi^*$ has continuous distribution.
\end{proof}

Note that in the special case when the distribution has finite mean, then additional equivalent statements are that  $X^*$ with cdf $F$ is a solution to a portfolio maximization problem \eqref{problem} for some $V$ that is upper semicontinuous, increasing and consistent with concave order.
	
\begin{remark}
	The equivalence result above can be slightly generalized. On the one hand, the domain of the utility function can be shifted to any interval $(a,\infty)$, $a \in \R$, see Remark 3.9 in \cite{BTZ04}. On the other hand, dropping the assumptions that $F_{\xi^*}$ is strictly increasing on its domain and $F$ is continuous while maintaining $\lim_{x \searrow a} F(x) =0$, we only lose uniqueness of the optimization problems. For the utility function $u_{\xi^*}$ this  loses the continuous differentiability while the dual is no longer strictly convex (but continuously differentiable). In this case the required duality result is given by Theorem 3.2 in \cite{BTZ04}, (noting that $\bigl( {u}^{\prime }\bigr)^{-1}$ has to be replaced by the equivalent $-\tilde{u}'$ as the utility function is no longer differentiable).   
\end{remark}

\begin{remark} In their paper, \cite{jouini2001efficient}  define efficiency as being an optimal solution for an expected utility maximization, which is directly related to  our notion of being rationalizable by Expected Utility Theory. Our results in Theorem \ref{thD2} are consistent with their findings in the case of a finite probability space $\Omega$. 
\end{remark}

\subsection{Perfectly Cost-efficient Distributions and Optimal Payoffs \label{PCED}}

While in complete markets, all payoffs are hedgeable per definition, and all distributions are hedgeable with a perfectly cost-efficient claim, this is not the case in incomplete markets. Specifically, given an arbitrary cdf $F$, there may not exist a cost-efficient claim $X^* \in D(F)$ such that $X^*$ is hedgeable. The subset of ``perfectly cost-efficient distributions'' is important as Corollary~\ref{co3.17} and Theorem~\ref{thD2} show  that only the payoffs that are perfectly cost-efficient are optimal for some portfolio choice problem as in \eqref{problem}.

The characterization of perfectly cost-efficient payoffs may thus be useful to describe the set of distributions that are hedgeable. An efficient portfolio choice would then always be among these hedgeable claims. Such a characterization can be useful to solve expected utility maximization problems. We illustrate in \cite{BS24appendix} 
 how the characterization of perfectly cost-efficient claims allows us to solve an expected utility maximization problem in a simple incomplete market (trinomial model) and recover the results from Chapter 3 of \cite{DS06} obtained using duality.

Note however, that the implications of such a result for hedging claims are not easy to grasp as it is typically not enough to hedge in distribution. A good hedging strategy matches future cash flows of a target claim. Hedging the distribution of this target claim can be misleading when interpreted in the context of risk management. For example, a fire-insurance contract paying 100,000 euros with 10\% probability and a digital financial option paying 100,000 euros with 10\% probability have the same distribution but cannot be hedged in the same way (see \cite{bernard2014financial}).

\section{Conclusions}\label{S7}

In this paper we generalize results on cost-efficiency to general incomplete market settings. We extend the notion of cost-efficiency through relaxation in convex order and we introduce the notion of perfect cost-efficiency, i.e., cost-efficiency while maintaining the target distribution as the key notion. In this way preference optimizers of  increasing and upper semicontinuous preference functionals that enjoy the property to be diversification-loving (which we coined for this purpose) are perfectly cost-efficient. Additionally, perfectly cost-efficient payoffs are shown, under mild regularity assumptions, to be rationalizable in the expected utility framework, i.e., they are optimizers for an (incomplete market) expected utility maximization problem. Similar to maximizing the expected utility in an incomplete market via duality, the results hinge on the optimality of pricing kernels that depend on the distribution targeted (as the pricing kernel involved in the optimum depends on the utility function). In the companion paper \cite{BS24appendix} we provide a rich array of examples and counterexamples that highlight our findings and illustrate that the notions introduced are sharp, i.e., they cover all reasonable cases and without them the results hold no longer true. We hope that this conceptual clarification on cost-efficiency in incomplete markets will enable us to reproduce and extend many of the results in complete markets relying on cost-efficiency now, properly qualified, to incomplete markets. 

 \section*{Acknowledgments}
 
We would like to thank an associate editor for very detailed and constructive comments that helped us revise our paper. All remaining errors are ours. We thank Alfred M\"uller, Moris Strub and Ruodu Wang for their feedback on an earlier version of this paper. Also thanks to the Institute for Pure \& Applied Mathematics (IPAM), where at the Long Program on Broad Perspectives and New Directions in Financial Mathematics the seeds for this collaboration were laid.  C. Bernard acknowledges funding from  
FWO G093024N at VUB.

\bibliographystyle{plainnat}
\bibliography{ICE_bib}
\newpage

\appendix

\begin{center}
{\LARGE \textbf{Appendix}}
\end{center}	

\section{Additional Information on $\convc(F)$ \label{appconvF}}

A key element in the analysis is the closed convex hull of a distribution, $\convc(F)$ which extends the notion of convex domination by $F$ also to distributions with infinite mean. As this is a quite curious object by itself, we want to present some results that help to understand better this object, in particular in the case of infinite mean.

\begin{lemma}\label{integrability}
	Let $F$ be a distribution supported on the non-negative reals $\R_{\geq 0}$. $F$ has a finite mean if and only if the convex closure $\convc(F)$ is bounded in probability.
\end{lemma}

\begin{proof}
	The forward direction is a direct consequence of Markov's inequality. Assume that $F$ has a finite mean, then
	\begin{align*}
		\lim_{M \to \infty} \sup_{X \in \convc(F)} \P\bigl[ \vert X \vert \geq M \bigr] & \leq  \lim_{M \to \infty} \sup_{X \in \convc(F)} \frac{\E\bigl[ \vert X \vert \bigr]}{M} =  \lim_{M \to \infty} \sup_{\substack{X_i \in D(F)}} \sup_{\substack{\lambda_i \geq 0, \, n \geq 1 \\ \sum_{i=1}^n \lambda_i = 1}} \frac{\E\bigl[ \vert \sum_{j=1}^n \lambda_j X_j \vert \bigr]}{M} \\
		& =  \lim_{M \to \infty} \sup_{\substack{X_i \in D(F)}} \sup_{\substack{\lambda_i \geq 0, \, n \geq 1 \\ \sum_{i=1}^n \lambda_i = 1}} \sum_{j=1}^n \lambda_j \frac{\E\bigl[ X_j  \bigr]}{M} = \lim_{M \to \infty} \frac{1}{M} \int_0^\infty  x \, dF(x) = 0,
	\end{align*}
	which proves boundedness in probability.
	
	For the backward direction, assume that $F$ has infinite mean. Let us recall that if a convex family of random variables in $\lnpofp$ is bounded in probability, there exists a measure $\Q$ equivalent to $\P$ such that the family is uniformly integrable under $\Q$ (see Lemma 2.3 (3) in \cite{BS99}). Denoting by $F^-$ the left continuous version of $F$, we have moreover $\frac{d\Q}{d\P} \geq \bigl(F^{-1}_{\frac{d\Q}{d\P}} \circ F^- \bigr)(Y)$, and thus
	\begin{align*}
		\E^\Q[Y] &= \E \biggl[ Y \frac{d\Q}{d\P} \biggr] \geq \E \Bigl[ Y \Bigl(F^{-1}_{\frac{d\Q}{d\P}} \circ F^- \Bigr)(Y) \Bigr] = \int_{[0,\infty)} y \Bigl(F^{-1}_{\frac{d\Q}{d\P}} \circ F^- \Bigr)(y) \, dF(y) \\
		& \geq  \int_{\bigl(F^{-1} \circ F_{\frac{d\Q}{d\P}}\bigr)(1,\infty\bigr)} y \Bigl(F^{-1}_{\frac{d\Q}{d\P}} \circ F^- \Bigr)(y) \, dF(y) \geq \int_{\bigl(F^{-1} \circ F_{\frac{d\Q}{d\P}}\bigr)(1,\infty\bigr)} y  \, dF(y) = \infty.
	\end{align*}
	Thus, $Y$ is not $\Q$-integrable and hence $\convc(F)$ cannot be bounded in probability for $F$ with infinite mean.
\end{proof}

A slightly weaker but more intuitive statement for the converse direction of Lemma~\ref{integrability} can be obtained when we recast the problem in terms of optimal couplings. To do so, we build on recent results of \cite{WPY} on joint mixability. We have the following proposition.

\begin{proposition} \label{propW}
	Let  $F$ be a cdf that has a decreasing density, then Lemma 2.2 of \cite{WPY} shows that there exists an optimal coupling $\tilde{X}_1, \ldots, \tilde{X}_n \in D(F)$ with $\tilde{S}_n = \tilde{X}_1 + \cdots+  \tilde{X}_n$ such that
	\begin{align*}
		m_+^n(s)  & := \inf\Bigl\{\P[S_n<s] \, : \, X_i \in D(F), \, \sum_{i=1}^n X_i = S_n \Bigr\} \\
		& = 1- \sup\Bigl\{\P[S_n \geq s] \, : \, X_i \in D(F) , \, \sum_{i=1}^n X_i = S_n\Bigr\}= \P\bigl[ \tilde{S}_n <s\bigr].
	\end{align*}
	Assuming that $F$ satisfies $\lim_{n \to \infty} \frac{1}{n} F^{-1}(1-\frac{1}{n}) = \infty$, then, for every $M >0$
	\[
		\lim_{n \to \infty} \P\bigl[ \tilde{S}_n  \geq nM\bigr] = 1.
	\]
\end{proposition}

This statement is quite remarkable: Given very heavily tailed random variables, the power of the correlation structure is so strong that the sum of the optimal coupling can almost surely surpass any linearly growing barrier. Or put it differently: the average of optimally coupled identically distributed random variables $\frac{\tilde{S}_n}{n}$ diverges to infinity with probability one. To prove this proposition, we make use of recent results on joint mixability of \cite{WPY}.

\begin{proof}
	From Theorem 3.4 of \cite{WPY}, the bound can be calculated directly as $m_+^n(s) = \phi_n^{-1}(s),$ in which $\phi_n$ is explicitly known. In the case when $F$ has a decreasing density, $\phi_n(s) := H_s\bigl(c_n(s)\bigr)$ with
	\[
		H_s(x) = (n-1) F^{-1}\bigl(s + (n-1)x\bigr) + F^{-1}(1-x)
	\]
	for $s \in [0,1]$ and $x\in\bigl(0,\frac{1-s}{n}\bigr]$, and 
	\[
		c_n(s) = \inf\biggl\{ c \in \biggl(0,\frac{1-s}{n}\biggr] \, : \, \int_c^\frac{1-s}{n} H_s(t)\, dt \geq \biggl(\frac{1-s}{n}-c\biggr)H_s(c)\biggr\}.
	\]
	Note that the infimum should be taken over $\bigl(0,\frac{1-s}{n}\bigr]$ instead of $\bigl[0,\frac{1-s}{n}\bigr]$ as written in the original paper, see Formula (3.7) in \cite{WPY}. We note first that for arbitrary $\varepsilon \in (0,1)$ as $c_n(1-\varepsilon) \leq \frac{1}{n}$ we have
	\begin{align*}
		\lim_{n \to \infty} \frac{\phi_n(1-\varepsilon)}{n} & = \lim_{n \to \infty} \Bigl(1-\frac{1}{n}\Bigr) F^{-1}\Bigl(1 + (n-1)c_n(1-\varepsilon)\Bigr) + \frac{1}{n}F^{-1}\Bigl(1-c_n(1-\varepsilon)\Bigr)\\
		&  \geq \lim_{n \to \infty} \frac{1}{n} F^{-1}\Bigl(1-\frac{1}{n}\Bigr) = \infty.
	\end{align*}
	Therefore, as $\phi_n$ is a strictly increasing function, we have for every $M > 0$ that
	\[
		\lim_{n \to \infty} \phi_n^{-1}(nM) <1 - \varepsilon
	\]
	and thus
	\[
		\lim_{n \to \infty} \P\bigl[\tilde{S}_n \geq nM\bigr] = 1- \lim_{n \to \infty} \phi_n^{-1}(nM) > \varepsilon.
	\]
	We can conclude that
	\[
		\lim_{M \to \infty} \sup_{X \in \convc(F)} \P\bigl[ \vert X \vert \geq M \bigr]  \geq 	\lim_{M \to \infty} \lim_{n \to \infty}  \P\bigl[\tilde{S}_n \geq nM\bigr]  > \varepsilon.
	\]
	Finally, as $\varepsilon$ was chosen arbitrarily in $(0,1)$, we can conclude that the limit is even $1$.
\end{proof}

An obvious example for a distribution satisfying the conditions of Proposition~\ref{propW} is the folded Cauchy distribution with density $f(x) = \frac{2}{\pi(1+x^2)}\ind_{\{x>0\}}$. Note however that the statement is strictly weaker than that of Lemma~\ref{integrability}, not only since we have the additional condition that the density is decreasing, but also because there are distributions that fail the condition $\lim_{n \to \infty} \frac{1}{n} F^{-1}(1-\frac{1}{n}) = \infty$, but nevertheless do not have finite mean. As example one might consider $f(x) = \frac{1}{(1+x)^2} \ind_{\{x >0\}}$. The reason for this gap is that while considering the optimal coupling we allow only for equally weighted averages, while in the convex closure we allow for arbitrarily weighted ones.

\end{document}